\documentclass[fleqn,11pt,twoside]{article}

\usepackage{amsthm,amsthm,amssymb, color, xcolor,epsfig, graphics, subfigure}

\usepackage{amsmath, graphicx, latexsym, lscape}

\usepackage{tikz,tikz-3dplot}
\newcommand{\bse}{\begin{subequations}}
\newcommand{\ese}{\end{subequations}}

\makeatletter
\newcommand{\copyrightnote}[2]{{\renewcommand{\thefootnote}{}
 \footnotetext{\small\it
\begin{flushleft}
 \copyright \ #1   #2  
\end{flushleft}}}}

\newcommand{\Name}[1]{\begin{flushleft}
                       \LARGE \bf #1
                       \end{flushleft}\vspace{-3mm}}

\newcommand{\Author}[1]{\begin{flushleft}
                       \it #1 \end{flushleft}}

\newcommand{\Address}[1]{\begin{flushleft}
                       \it #1 \end{flushleft}}

\newcommand{\Date}[1]{\begin{flushleft}
                      \small  \it #1 \end{flushleft}}

\newcommand{\evenhead}{Author \ name}
\newcommand{\oddhead}{Article \ name}

\renewcommand{\@evenhead}{
\hspace*{-3pt}\raisebox{-15pt}[\headheight][0pt]{\vbox{\hbox to \textwidth
{\thepage \hfil \evenhead}\vskip4pt \hrule}}}
\renewcommand{\@oddhead}{
\hspace*{-3pt}\raisebox{-15pt}[\headheight][0pt]{\vbox{\hbox to \textwidth
{\oddhead \hfil \thepage}\vskip4pt\hrule}}}
\renewcommand{\@evenfoot}{}
\renewcommand{\@oddfoot}{}

\long\def\@makecaption#1#2{%
  \vskip\abovecaptionskip
  \sbox\@tempboxa{\small \textbf{#1.}\ \ #2}%
  \ifdim \wd\@tempboxa >\hsize
    {\small \textbf{#1.}\ \ #2}\par
  \else
    \global \@minipagefalse
    \hb@xt@\hsize{\hfil\box\@tempboxa\hfil}%
  \fi
  \vskip\belowcaptionskip}

\newcommand{\JNMPnumberwithin}[3][\arabic]{%
  \@ifundefined{c@#2}{\@nocounterr{#2}}{%
    \@ifundefined{c@#3}{\@nocnterr{#3}}{%
      \@addtoreset{#2}{#3}%
      \@xp\xdef\csname the#2\endcsname{%
        \@xp\@nx\csname the#3\endcsname .\@nx#1{#2}}}}%
}

\renewenvironment{proof}[1][\proofname]{\par
  \normalfont
  \topsep6\p@\@plus6\p@ \trivlist
  \item[\hskip\labelsep\textbf{%
    #1\@addpunct{.}}]\ignorespaces
}{%
  \qed\endtrivlist
}

\newcommand{\resetfootnoterule} {
  \renewcommand\footnoterule{%
  \kern-3\p@
  \hrule\@width.4\columnwidth
  \kern2.6\p@}
}

\renewcommand{\footnoterule}{}

\makeatother

\theoremstyle{definition}

\newtheorem{theorem}{Theorem}
\newtheorem{prop}{Proposition}
\newtheorem{cor}{Corollary}
\newtheorem{defn}{Definition}

\begin{document}

\renewcommand{\evenhead}{ {\LARGE\textcolor{blue!10!black!40!green}{{\sf \ \ \ ]ocnmp[}}}\strut\hfill 
P Jennings and F Nijhoff
}
\renewcommand{\oddhead}{ {\LARGE\textcolor{blue!10!black!40!green}{{\sf ]ocnmp[}}}\ \ \ \ \  
Hyperspherical Trigonometry and Corresponding Elliptic Functions
}

%%%% Matter for the first page 
\thispagestyle{empty}
\newcommand{\FistPageHead}[3]{
\begin{flushleft}
\raisebox{8mm}[0pt][0pt]
{\footnotesize \sf
\parbox{150mm}{{\textcolor{blue!10!black!40!green}{{\bf Open Communications in Nonlinear Mathematical Physics}}}
\ \ {Special Issue: Hietarinta}, 2026\\[0.1cm]
\strut\hfill 
ocnmp:18230
pp #2\hfill {\sc #3}}}\vspace{-13mm}
\end{flushleft}}

\FistPageHead{1}{\pageref{firstpage}--\pageref{lastpage}}{ \ \ }

\strut\hfill

\strut\hfill

\copyrightnote{The author. Distributed under a Creative Commons Attribution 4.0 International License}

\begin{center}

{\bf {\large A Special OCNMP Issue in Honour of Jarmo Hietarinta}}\\[0.2cm]
{\bf {\large on the Occasion of his 80th Birthday}}
\end{center}

\smallskip

\Name{Hyperspherical Trigonometry and Corresponding Elliptic Functions}

\Author{Paul Jennings}

\Address{Current address withheld}

\Author{Frank Nijhoff} 

\Address{School of Mathematics, University of Leeds, 
Leeds LS2 9JT, UK} 

\Date{Received May 19, 2026; Accepted June 4, 2026}

\setcounter{equation}{0}

\smallskip

\noindent
{\bf Citation format for this Article:}\newline
P Jennings and F Nijhoff, Hyperspherical Trigonometry and Corresponding Elliptic Functions,
{\it Open Commun. Nonlinear Math. Phys.}, Special Issue:\,Hietarinta, ocnmp:18230, \pageref{firstpage}--\pageref{lastpage}, 2026.

\strut\hfill

\noindent
{\bf The permanent Digital Object Identifier (DOI) for this Article:}\newline
{\it 10.46298/ocnmp.18230}
\strut\hfill

\begin{abstract}

\noindent 
We develop the basic formulae of hyperspherical trigonometry in multidimensional Euclidean space, using multidimensional vector products, and their conversion to identities for elliptic functions. We show that the basic addition formulae for functions on the 3-sphere embedded in 4-dimensional space lead to addition formulae for elliptic functions, associated with algebraic curves, which have two distinct moduli. We give an application of these formulae to the cases of a multidimensional Euler top, using them to provide a link to the Double Elliptic model.

\end{abstract}

\label{firstpage}

%%%% The Article text starts here

\section{Introduction}

Spherical trigonometry, a branch of geometry dealing with the goniometry of the angles of triangles confined to the surface of a 2-
sphere, is an area of mathematics that has existed since the ancient Greeks, with its foundations laid by Menelaus and Hipparchus. It 
is of vital importance for many calculations in astronomy, navigation and cartography. Further advances made in the Islamic world, in 
order to help calculate its Holy Days based on the phases of the moon, resulted in giving us the basis of spherical trigonometry in its 
modern form. In the western world in the 17th century spherical trigonometry was first considered as a separate mathematical 
discipline, independent of astronomy. One of the protagonists of this era, John Napier, in his work of 1614, treated spherical 
trigonometry alongside his work introducing logarithms. Other protagonists include Delambre, Euler, Cagnoli and l'Huillier.

For obvious reasons the trigonometry of hyperspheres in higher dimensions was not studied so intently. McMahon produced a number of 
formulae as generalisations of some of those from spherical trigonometry, \cite{2}. More recently, Sato considered the relationship 
between the dihedral angles of spherical simplices and those of their polars, \cite{sato}. There also exists a significant amount of work 
looking into the `Law of Sines', generalisations of the sine rule from the spherical case, \cite{eriksson}. Various formulae have also 
arisen in the the work of Derevnin, Mednykh and Pashkevich in their work on spherical volumes, \cite{Derevnin,Derevnin2}.  Recently, it 
has been shown by Petrera and Suris that the cosine rule for spherical triangles and tetrahedra define integrable systems, \cite{Suris}. 

It is well known that elliptic functions are related to spherical trigonometry through their addition formulae, \cite{Lagrange,Legendre}. These functions are commonly defined through the inversion of integrals, and they consequently obey differential equations associated with certain algebraic curves. A great deal of research took place in this area in the 19th century, comprising  works by many of the great mathematicians , including Euler, Jacobi, Legendre and Frobenius. As for higher dimensional hyperspherical trigonometry no such connection is known. It may be expected that this link would be through higher genus elliptic curves, for example Abelian functions, \cite{baker2}. However, this is not in fact the case. We show that instead the link is through the  `Generalised Jacobi Elliptic Functions' explicitly defined by Pawellek, \cite{pawellek}, and building upon Jacobi's work in this area,\cite{Jacobi}, as an elliptic covering, dependent on two distinct moduli.

In this paper we develop a complete set of formulae for hyperspherical trigonometry, and explore their link with elliptic 
functions. We establish a novel connection between the generalised Jacobi elliptic functions and the formulae of hyperspherical trigonometry. We show that through this connection the basic addition formulae of hyperspherical trigonometry lead to addition formulae for these generalised Jacobi elliptic functions. We apply this connection to an example, the four-dimensional Euler top. We then show that this system is equivalent to the one-particle Double Elliptic model.

Our interest in this subject area arises in relation to the tetrahedron equation, a higher dimensional generalisation of the Yang-Baxter equation, \cite{Zamolodchikov}. Zamolodchikov presented an intuitive solution to the latter dependent on spherical trigonometry, cf. also \cite{Baxter}. It is therefore natural for higher generalisations to be solved in terms of higher dimensional hyperspherical trigonometry.

The outline of this paper is as follows. In section 2 we review multidimensional vector products as a higher dimensional analogue of the standard cross-product of vectors, which we need to obtain the formulae of hyperspherical trigonometry. In section 3 we provide a summary of the well-known formulae of spherical trigonometry. We deduce analogous formulae for the four-dimensional hyperspherical case, and  using the same principles, do the same for the general $n$-dimensional case.  In section 4 we review the link between spherical trigonometry and the Jacobi elliptic functions, and generalise this to the link between four-dimensional hyperspherical trigonometry and elliptic functions. Section 5 is a brief introduction to the generalised Jacobi elliptic functions, \cite{pawellek}, complete with a link between the formulae of hyperspherical trigonometry and these functions. We conclude with two examples, the first being the four-dimensional Euler top, and the second the double elliptic (DELL) model, providing a connection between the two.

\noindent{\bf Note:} This paper summarises some collaborative results which 
appeared in the PhD thesis of one of the authors, \cite{Jennings-thesis}.

\section{Multidimensional Vector Products}
\label{sec:multidim products}

\subsection{Definition and Relations}
The formulae involved in spherical trigonometry are dependent on the cross product between vectors. The vector product $\mathbf{a}\times\mathbf{b}$ in 3D Euclidean space is a binary operation defined by
\begin{equation}
(\mathbf{a}\times\mathbf{b})_i=\mathrm{det}(\mathbf{a},\mathbf{b},\mathbf{e}_i),\hspace{20mm}(i=1,2,3),
\end{equation}
with  $\mathbf{e}_1$, $\mathbf{e}_2$, $\mathbf{e}_3$ the standard unit vectors in the orthogonal basis. This product obeys the rules:
\begin{itemize}
	\item Anti-commutative: $\mathbf{a}\times\mathbf{b}=-\mathbf{b}\times\mathbf{a}.$
	\item Vector Triple Product: $(\mathbf{a}\times\mathbf{b})\times\mathbf{c}=(\mathbf{a}\cdot\mathbf{c})\mathbf{b}-(\mathbf{b}\cdot\mathbf{c})\mathbf{a}=-\left|
	\begin{array}{cc}
	\mathbf{a} & \mathbf{b} \\
	\mathbf{a}\cdot\mathbf{c} & \mathbf{b}\cdot\mathbf{c}\\	
	\end{array}\right|,$
	
	from which it follows
	$(\mathbf{a}\times\mathbf{b})\times\mathbf{c}-(\mathbf{a}\times\mathbf{c})\times\mathbf{b}=\mathbf{a}\times(\mathbf{b}\times\mathbf{c}).$
	\item Area of a Parallelogram: The modulus of the vector product, $|\mathbf{a}\times\mathbf{b}|,$ is equivalent to the area of the parallelogram defined by these vectors, $|\mathbf{a}\times\mathbf{b}|=\sin\theta,$ with $\theta$ the obtuse angle between them $(0\leq\theta\leq\pi).$
\end{itemize}
For higher dimensional spherical trigonometry an $n$-ary operation between $n$ vectors is required. A natural vector product in four dimensions will therefore be a ternary vector product, $\mathbf{a}\times\mathbf{b}\times\mathbf{c}$, of three vectors in 4D Euclidean space $E_4$, defined in a similar manner, by
\begin{equation}
(\mathbf{a}\times\mathbf{b}\times\mathbf{c})\cdot\mathbf{d}=\mbox{det}(\mathbf{a},\mathbf{b},\mathbf{c},\mathbf{d}),\hspace{5mm}\mathrm{for\,\,all\,\, vectors}\,\mathbf{d}\in E_4.
\end{equation}
More generally, this could be extended to an $n$-dimensional vector product as an $n$-ary operation of $n-1$ vectors in $E_n$, defined by the expression, \cite{Grassman,Acosta}, 
\begin{equation}
(\mathbf{a}_1\times\mathbf{a}_2\times\dots\times\mathbf{a}_{n-1})\cdot\mathbf{a}_n=\mbox{det}(\mathbf{a}_1,\mathbf{a}_2,\dots,\mathbf{a}_{n-1},\mathbf{a}_n).
\end{equation}
It clearly follows from this definition that the multiple vector products are antisymmetric with respect to the interchangement of their constituent vectors. Furthermore, these multiple vector products are perpendicular to any one of their constituent vectors. In fact, the multidimensional vector product 
follows directly from the wedge product of the exterior algebra, 
\cite{Janich}. 

The following nested product identity involving five vectors holds for the triple vector product in $E_4=\mathbb{R}^{4}$,
\begin{equation}
(\mathbf{a}\times\mathbf{b}\times\mathbf{c})\times\mathbf{d}\times\mathbf{e}=-
\left|
\begin{array}{ccc}
\mathbf{a}&
\mathbf{b}&
\mathbf{c}\\
(\mathbf{a}\cdot\mathbf{d})&
(\mathbf{b}\cdot\mathbf{d})&
(\mathbf{c}\cdot\mathbf{d})\\
(\mathbf{a}\cdot\mathbf{e})&
(\mathbf{b}\cdot\mathbf{e})&
(\mathbf{c}\cdot\mathbf{e})\\
\end{array}
\right|.
\end{equation}
This follows from the more general higher-dimensional analogue. 

\begin{prop} {\rm (Nested Vector Product Identity)}. 
For 2n-1 vectors, $\mathbf{a}_i\in \mathbb{R}^{n+1}$, $i=1,\dots,2n-1$, we have the identity 
\begin{equation}
(\mathbf{a}_1\times\mathbf{a}_2\times\dots\times\mathbf{a}_n)\times\mathbf{a}_{n+1}\times\dots\times\mathbf{a}_{2n-1}=-
\left|
\begin{array}{cccc}
\mathbf{a}_1&
\mathbf{a}_2&
\dots&
\mathbf{a}_n\\
(\mathbf{a}_1\cdot\mathbf{a}_{n+1})&
(\mathbf{a}_2\cdot\mathbf{a}_{n+1})&
\dots&
(\mathbf{a}_n\cdot\mathbf{a}_{n+1})\\
\vdots&
\vdots&
\ddots&
\vdots\\
(\mathbf{a}_1\cdot\mathbf{a}_{2n-1})&
(\mathbf{a}_2\cdot\mathbf{a}_{2n-1})&
\dots&
(\mathbf{a}_n\cdot\mathbf{a}_{2n-1})\\
\end{array}
\right|.
\end{equation}\end{prop} 

\begin{proof}
Consider the determinantal expression
\begin{equation}
\mbox{det}(\mathbf{a}_1,\dots\mathbf{a}_{n+1})=\sum_{i_1}\dots\sum_{i_{n+1}}\varepsilon_{i_1\dots i_{n+1}}(\mathbf{a}_1)_{i_1}\dots(\mathbf{a}_{n+1})_{i_{n+1}},
\end{equation}
where $\varepsilon_{i_1\dots i_{n+1}}$ is the $(n+1)$-dimensional Levi-Civita symbol. From this it follows that
\begin{equation}
\begin{split}
[(\mathbf{a}_1&\times\mathbf{a}_2\times\dots\times\mathbf{a}_n)\times\mathbf{a}_{n+1}\times\dots\times\mathbf{a}_{2n-1}]_{j_n}\\
&=\sum_{i_1}\dots\sum_{i_{n+1}}\sum_{j_1}\dots\sum_{j_{n}}\epsilon_{j_1i_1\dots i_n}\epsilon_{j_1\dots j_{n+1}}(\mathbf{a}_1)_{i_1}\dots(\mathbf{a}_{n})_{i_{n}}(\mathbf{a}_{n+1})_{j_2}\dots(\mathbf{a}_{2n-1})_{j_{n}}.
\end{split}
\end{equation}
Noting that the Levi-Civita symbol satisfies the following product rule
\begin{equation}
\sum_{j_1}\epsilon_{j_1i_1\dots i_n}\epsilon_{j_1\dots j_{n+1}}=\left|
\begin{array}{ccc}
\delta_{i_1j_2}&\dots & \delta_{i_1j_{n+1}}\\
\vdots & &\vdots\\
\delta_{i_nj_2}&\dots & \delta_{i_{n}j_{n+1}}\\
\end{array}
\right|,
\end{equation}
the result follows.
\end{proof}

Vectorial addition identities follow from the Pl\"ucker relations in projective geometry. The Pl\"ucker relations are identities involving minors of non-square matrices which are the Pl\"ucker coordinates of corresponding Grassmannians.

\begin{prop}[Pl\"ucker Relations]
For $(2n-2)$ vectors $\mathbf{a}_1,\mathbf{a}_2,\dots,\mathbf{a}_{2n-2}\in\mathbb{R}^{n-1}$,
\begin{equation}
\label{eq:plucker}
\begin{split}
(\mathbf{a}_1,\mathbf{a}_{n+1},\dots,\mathbf{a}_{2n-2})&(\mathbf{a}_2,\dots,\mathbf{a}_n)-(\mathbf{a}_2,\mathbf{a}_{n+1},\dots,\mathbf{a}_{2n-2})(\mathbf{a}_1,\mathbf{a}_3,\dots,\mathbf{a}_n)\\
&+\dots+(-1)^{n-1}(\mathbf{a}_n,\mathbf{a}_{n+1},\dots,\mathbf{a}_{2n-2})(\mathbf{a}_1,\dots,\mathbf{a}_{n-1})=0.
\end{split}
\end{equation}
\end{prop}

\begin{proof}
Consider $n$ vectors, $\mathbf{a}_1, \mathbf{a}_2,\dots,\mathbf{a}_n$ in $\mathbb{R}^{n-1}$. As there are $n$ vectors in $(n-1)$ dimensional space they must be linearly dependent, hence,
\begin{equation}
\left|\begin{array}{cccc}
(\mathbf{a}_1)&(\mathbf{a}_2)&\cdots&(\mathbf{a}_n)\\
\mathbf{a}_1&\mathbf{a}_2& \cdots &\mathbf{a}_n\\
\end{array}\right|=0,
\end{equation}
where $(\mathbf{a}_i)$ represents the column vector of the components of $\mathbf{a}_i$, $i\in\{1,2,\dots,n\}$. This implies
\begin{equation}
\mathbf{a}_1(\mathbf{a}_2,\dots,\mathbf{a}_n)-\mathbf{a_2}(\mathbf{a}_1,\mathbf{a}_3,\dots,\mathbf{a}_n)+\dots+(-1)^{n-1}\mathbf{a}_n(\mathbf{a}_1,\dots,\mathbf{a}_{n-1})=0,
\end{equation}
which, in turn gives
\begin{equation}
\begin{split}
(\mathbf{a}_1,\mathbf{a}_{n+1},\dots,\mathbf{a}_{2n-2})&(\mathbf{a}_2,\dots,\mathbf{a}_n)-(\mathbf{a}_2,\mathbf{a}_{n+1},\dots,\mathbf{a}_{2n-2})(\mathbf{a}_1,\mathbf{a}_3,\dots,\mathbf{a}_n)\\
&+\dots+(-1)^{n-1}(\mathbf{a}_n,\mathbf{a}_{n+1},\dots,\mathbf{a}_{2n-2})(\mathbf{a}_1,\dots,\mathbf{a}_{n-1})=0,
\end{split}
\end{equation}
for some arbitrary $\mathbf{a}_{n+1},\dots, \mathbf{a}_{2n-2}\in\mathbb{R}^{n-1}$, as required. 
\end{proof}

\begin{cor}
For vectors $\mathbf{a}_1, \dots, \mathbf{a}_{2n-2}\in\mathbb{R}^{n-1}$,
\begin{equation}
\begin{split}
&(\mathbf{a}_1\times\dots\times\mathbf{a}_n)\times\mathbf{a}_{n+1}\times\dots\times\mathbf{a}_{2n-1}\\
&+\mathbf{a}_n\times(\mathbf{a}_1\times\dots\times\mathbf{a}_{n-1}\times\mathbf{a}_{n+1})\times\mathbf{a}_{n+2}\times\dots\times\mathbf{a}_{2n-1}+\dots\\
&+\mathbf{a}_n\times\dots\times\mathbf{a}_{2n-2}\times(\mathbf{a}_1\times\dots\times\mathbf{a}_{n-1}\times\mathbf{a}_{2n-1})\\
&=\mathbf{a}_1\times\dots\times\mathbf{a}_{n-1}\times(\mathbf{a}_{n}\times\dots\times\mathbf{a}_{2n-1}).
\end{split}
\end{equation}
\end{cor}

\begin{proof}
From the Pl\"uker relations, (\ref{eq:plucker}), set $\mathbf{a}_{2n-2}=\mathbf{a}_n=\mathbf{e}_i$ and sum over $i$. This implies
\begin{equation}
\begin{split}
&(\mathbf{a}_1\times\mathbf{a}_{n+1}\times\dots\times\mathbf{a}_{2n-3})\cdot(\mathbf{a}_2\times\dots\times\mathbf{a}_{n-1})\\
&-(\mathbf{a}_2\times\mathbf{a}_{n+1}\times\dots\times\mathbf{a}_{2n-3})\cdot(\mathbf{a}_1\times\mathbf{a}_3\dots\times\mathbf{a}_{n-1})\\
&+\dots+(-1)^n(\mathbf{a}_{n-1}\times\mathbf{a}_{n+1}\times\dots\times\mathbf{a}_{2n-3})\cdot(\mathbf{a}_1\times\dots\times\mathbf{a}_{n-2})\\
&=0,
\end{split}
\end{equation}
which, using $(\mathbf{a}_1\times\mathbf{a}_2\times\dots\mathbf{a}_{n-2})\cdot\mathbf{a}_{n-1}=-(\mathbf{a}_2\times\dots\times\mathbf{a}_{n-1})\cdot\mathbf{a}_1$, implies
\begin{equation}
\begin{split}
&\mathbf{a}_1\times\mathbf{a}_{n+1}\times\dots\times\mathbf{a}_{2n-4}\times(\mathbf{a_2}\times\dots\times\mathbf{a}_{n-1})\\
&-\mathbf{a}_2\times\mathbf{a}_{n+1}\times\dots\times\mathbf{a}_{2n-4}\times(\mathbf{a_1}\times\mathbf{a}_{3}\times\dots\times\mathbf{a}_{n-1})\\
&+\dots+(-1)^n\mathbf{a}_{n-1}\times\mathbf{a}_{n+1}\times\dots\times\mathbf{a}_{2n-4}\times(\mathbf{a_1}\times\dots\times\mathbf{a}_{n-2})\\
&=0.
\end{split}
\end{equation}
Summing over $(n-2)$ copies, the result follows.
\end{proof}

\begin{cor}
More specifically, for $\mathbf{a}, \mathbf{b}, \mathbf{c}, \mathbf{d}, \mathbf{e}\in\mathbb{R}^4$,
\begin{equation}\label{9}
(\mathbf{a}\times\mathbf{b}\times\mathbf{c})\times\mathbf{d}\times\mathbf{e}+(\mathbf{a}\times\mathbf{b}\times\mathbf{d})\times\mathbf{e}\times\mathbf{c}+(\mathbf{a}\times\mathbf{b}\times\mathbf{e})\times\mathbf{c}\times\mathbf{d}=\mathbf{a}\times\mathbf{b}\times(\mathbf{c}\times\mathbf{d}\times\mathbf{e}).
\end{equation}
\end{cor}
This identity will be particularly important in proving the various hyperspherical identities in the four-dimensional case.

\section{Hyperspherical Trigonometry}

In  this section we review the formulae for spherical trigonometry as a preparation for developing similar formulae in higher dimensions. The approach we take is by exploiting the higher dimensional vector product introduced in section~\ref{sec:multidim products}.
\subsection{Spherical Trigonometry}
We review the derivation of the basic formulae of spherical trigonometry, cf. also \cite{1}. Consider a big spherical triangle on the surface of a $2$-sphere of unit radius embedded in three dimensional Euclidean space, that is a triangle bounded by three big circles on $S^2$, restricted such that each side is less than a semicircle, and hence, each angle is less than $\pi$. Take the centre of the sphere to be the origin in $\mathbf{E}_3=\mathbb{R}^3$ and denote the position vectors of the three vertices of the spherical triangle by $\mathbf{n}_1,\mathbf{n}_2$ and $\mathbf{n}_3$, with the angles, $\theta_{ij}$, between them, corresponding to the edges, being defined by
\begin{equation}
\mathbf{n}_i\cdot\textbf{n}_j\equiv\cos\theta_{ij},\hspace{20 mm}i,j=1,2,3.
\end{equation}
Introduce the vectors
\begin{equation}
\textbf{u}_{ij}\equiv\frac
{\textbf{n}_i\times\textbf{n}_j}
{\left|\textbf{n}_i\times\textbf{n}_j\right|},
\end{equation}
and define the spherical angles, $\alpha_j$, between them by
\begin{equation}
\textbf{u}_{ij}\cdot\textbf{u}_{jk}\equiv-\cos\alpha_{j}.
\end{equation}

\begin{figure}[htb]
\centering
\tdplotsetmaincoords{60}{120}
\begin{tikzpicture}[tdplot_main_coords,scale=5]
\draw[thick,->](0,0,0) -- (1,0,0)node[anchor=north east]{$\mathbf{n}_k$};
\draw[thick,->](0,0,0) -- (0,1,0)node[anchor=west]{$\mathbf{n}_j$};
\draw[thick,->](0,0,0) -- (0,0,1)node[anchor=south]{$\mathbf{n}_i$};
\draw (0,0)circle (1 cm);
\tdplotdrawarc[thick]{(0,0,0)}{1}{0}{90}{}{}
\tdplotdrawarc{(0,0,0)}{0.1}{0}{90}{anchor=north}{$\theta_{jk}$}
\tdplotdrawarc[dashed]{(0,0,0)}{1}{90}{360}{}{}
\tdplotdrawarc{(0,0,1)}{0.1}{2}{86}{anchor=north}{$\alpha_{i}$}

\tdplotsetrotatedcoords{90}{90}{0}

\tdplotdrawarc[tdplot_rotated_coords,thick]{(0,0,0)}{1}{180}{270}{}{}
\tdplotdrawarc[tdplot_rotated_coords]{(0,0,0)}{0.1}{180}{270}{anchor=east}{$\theta_{ik}$}
\tdplotdrawarc[tdplot_rotated_coords,dashed]{(0,0,0)}{1}{0}{180}{}{}
\tdplotdrawarc[tdplot_rotated_coords,dashed]{(0,0,0)}{1}{270}{360}{}{}
\tdplotdrawarc[tdplot_rotated_coords]{(0,0,1)}{0.1}{184}{268}{anchor=east}{$\alpha_{j}$}

\tdplotsetrotatedcoords{0}{90}{90}
\tdplotdrawarc[tdplot_rotated_coords,thick]{(0,0,0)}{1}{0}{90}{}{}
\tdplotdrawarc[tdplot_rotated_coords]{(0,0,0)}{0.1}{0}{90}{anchor=south west}{$\theta_{ij}$}
\tdplotdrawarc[tdplot_rotated_coords,dashed]{(0,0,0)}{1}{90}{360}{}{}
\tdplotdrawarc[tdplot_rotated_coords]{(0,0,1)}{0.1}{0}{90}{anchor=west}{$\alpha_k$}

\end{tikzpicture}
\caption{A spherical triangle}
\label{fig:spherical_triangle}
\end{figure}
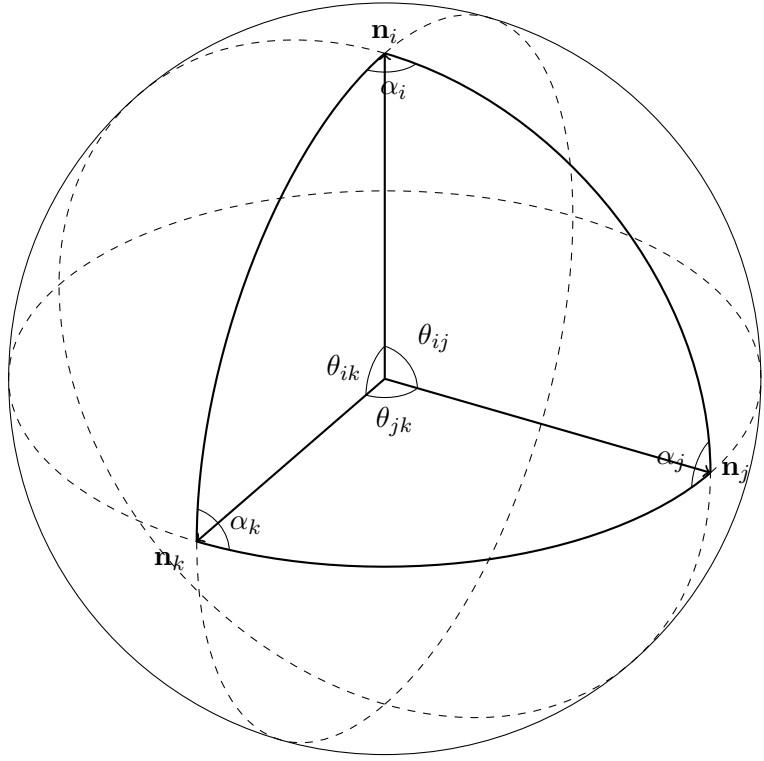

\begin{theorem}[Spherical Excess]
The area, $A$, of a spherical triangle is given by its $\mathbf{Spherical}$ $\mathbf{Excess}$,
\begin{equation} 
A= \alpha_i+\alpha_j+\alpha_k-\pi.
\end{equation}
\end{theorem}

\begin{proof}
Scl\"afli's Differential Volume formula states that the area of a spherical triangle must satisfy
\begin{equation}
\mathrm{d}A=\sum_{i=1}^3\mathrm{d}\alpha_i.
\end{equation}
Integrating this gives
\begin{equation}
A=\alpha_i+\alpha_j+\alpha_k+c,
\end{equation}
where $c=-\pi$ is the integration constant determined by considering, say, the area of a spherical triangle with planar angles $\alpha_i=\alpha_j=\alpha_k=\frac{\pi}{2}$ and comparing this to the known area of $1/8$ of a sphere.
\end{proof}
This theorem may also be proven by considering the areas of intersecting lunars, \cite{1}.
\begin{defn}[Polar Triangle]
The spherical triangle defined by the vectors $\mathbf{u}_{ij}$, $\mathbf{u}_{ik}$, $\mathbf{u}_{jk}$ is called the $\mathbf{polar}$ of the spherical triangle defined by the vectors $\mathbf{n}_{i}$, $\mathbf{n}_{j}$, $\mathbf{n}_{k}$.
\end{defn}
By considering various relations for the scalar and vector products between the polar vectors $\textbf{u}_{ij}$, we can define a number of important relations between the angles of a spherical triangle.
\begin{prop}[Cosine Rule]
\begin{equation}\label{eq:23}
\cos\alpha_{j}=\frac{\cos\theta_{ik}-\cos\theta_{ij}\cos\theta_{jk}}{\sin\theta_{ij}\sin\theta_{jk}},
\end{equation} for all $i,j,k=1,2,3$. 
\end{prop}

\begin{proof}
Consider the scalar product,
\begin{equation}
\begin{split}
\textbf{u}_{ij}\cdot\textbf{u}_{jk}
=&\frac
{(\textbf{n}_i\times\textbf{n}_j)\cdot(\textbf{n}_j\times\textbf{n}_k)}
{\left|\textbf{n}_i\times\textbf{n}_j\right|\left|\textbf{n}_j\times\textbf{n}_k\right|},\\
=&\frac{\cos\theta_{ij}\cos\theta_{jk}-\cos\theta_{ik}}
{\sin\theta_{ij}\sin\theta_{jk}}.
\end{split}
\end{equation}
Hence, the Cosine Rule follows.
\end{proof}

\begin{defn}
The \textbf{generalised} \textbf{sine} \textbf{function} \textbf{of} \textbf{three} \textbf{variables}, $\sin(\theta_{ij},\theta_{jk},\theta_{ik})$, is defined by
\begin{equation}
\begin{split}
\sin(\theta_{ij},\theta_{jk},\theta_{ik})&=\sqrt{1-\cos^2\theta_{ij}-\cos^2\theta_{jk}-\cos^2\theta_{ik}+2\cos\theta_{jk}\cos\theta_{ij}\cos\theta_{ik}},\\
&=\left|\begin{array}{ccc}
1 & \cos\theta_{ij} &\cos\theta_{ik} \\
\cos\theta_{ij} & 1 & \cos\theta_{jk} \\
\cos\theta_{ik} & \cos\theta_{jk} & 1 \\
\end{array}\right|^{\frac{1}{2}}, \hspace{10mm}(0\leq\theta_{ij},\theta_{ik},\theta_{jk}\leq\pi).
\end{split}
\end{equation}
Note that the restrictions for a spherical triangle, that each side is less than a semi-circle and each angle less than $\pi$, ensure that the generalised sine function is always real and positive.
\end{defn}

This sine function is equivalent to the volume of a parallelepiped with edges $\mathbf{n}_i$, $\mathbf{n}_j$ and $\mathbf{n}_k$, given by
\begin{equation}
\text{Volume}=|\textbf{n}_i\cdot(\textbf{n}_j\times\textbf{n}_k)|=\sin(\theta_{ij},\theta_{ik},\theta_{jk}).
\end{equation}
This identity can be viewed of as a consequence of the triple product when we embed the 3D vectors $\mathbf{n}_i$, $\mathbf{n}_j$ and $\mathbf{n}_k$ in $E_4$ by
\begin{equation}
\widetilde{\mathbf{n}}_i=\left(\begin{array}{c} \mathbf{n}_i \\ 0\end{array}\right)\ , \quad  
\widetilde{\mathbf{n}}_j=\left(\begin{array}{c} \mathbf{n}_j \\ 0\end{array}\right)\ , \quad
\widetilde{\mathbf{n}}_k=\left(\begin{array}{c} \mathbf{n}_k \\ 0\end{array}\right)\ ,
\end{equation}
such that 
\begin{equation}
\widetilde{\mathbf{n}}_i\times\widetilde{\mathbf{n}}_j\times\widetilde{\mathbf{n}}_k=\left(\begin{array}{c} 0 \\ 0\\ 0 \\ \det(\mathbf{n}_i,\mathbf{n}_j,\mathbf{n}_k)\end{array}\right)\  , 
\end{equation} 
where the last entry equals $|\mathbf{n}_i\cdot(\mathbf{n}_j\times\mathbf{n}_k)|$.

\begin{prop}[Sine Rule]
\begin{equation}
\frac{\sin\alpha_i}{\sin\theta_{jk}}=
\frac{\sin\alpha_j}{\sin\theta_{ik}}=
\frac{\sin\alpha_k}{\sin\theta_{ij}}=
\frac{\sin(\theta_{ij},\theta_{jk},\theta_{ik})}{\sin\theta_{ij}\sin\theta_{jk}\sin\theta_{ik}}=k,
\end{equation}
where $k$ is a constant.
\end{prop}

\begin{proof}
Consider the ratio
\begin{equation}\label{ratio}
\frac{|\textbf{u}_{ij}\times\textbf{u}_{jk}|}{|\textbf{n}_i\times\textbf{n}_k|}=\frac{\sin\alpha_j}{\sin\theta_{ik}}.
\end{equation}
Now, note the vector product
\begin{equation}\label{eq:29}
\textbf{u}_{ij}\times\textbf{u}_{jk}=\frac{(\textbf{n}_i\times\textbf{n}_j)\times(\textbf{n}_j\times\textbf{n}_k)}{|\textbf{n}_i\times\textbf{n}_j||\textbf{n}_j\times\textbf{n}_k|}
=\frac{\left(\textbf{n}_i\cdot(\textbf{n}_j\times\textbf{n}_k)\right)\textbf{n}_j}{\sin\theta_{ij}\sin\theta_{jk}}.
\end{equation}
This implies
\begin{equation}
|\textbf{u}_{ij}\times\textbf{u}_{jk}|=\left|\frac{\textbf{n}_i\cdot(\textbf{n}_j\times\textbf{n}_k)}{\sin\theta_{ij}\sin\theta_{jk}}\right|
=\frac{\sin(\theta_{ij},\theta_{ik},\theta_{jk})}{\sin\theta_{ij}\sin\theta_{jk}},
\end{equation}
and hence, the sine rule follows.
\end{proof}

\begin{prop}[Polar Cosine Rule]
\begin{equation}
\cos\theta_{jk}=\frac{\cos\alpha_{j}\cos\alpha_{k}+\cos\alpha_{i}}{\sin\alpha_j\sin\alpha_k},
\end{equation}for all $i,j,k=1,2,3$.
\end{prop}

\begin{proof}
Consider the product
\begin{equation}
\frac{(\textbf{u}_{ij}\times\textbf{u}_{jk})\cdot(\textbf{u}_{jk}\times\textbf{u}_{ik})}
{|\textbf{u}_{ij}\times\textbf{u}_{jk}||\textbf{u}_{jk}\times\textbf{u}_{ik}|}.
\end{equation}
This product can be calculated in two ways. First,
\begin{equation}
\frac{(\textbf{u}_{ij}\times\textbf{u}_{jk})\cdot(\textbf{u}_{jk}\times\textbf{u}_{ik})}
{|\textbf{u}_{ij}\times\textbf{u}_{jk}||\textbf{u}_{jk}\times\textbf{u}_{ik}|}=\frac{-\cos\alpha_j\cos\alpha_k-\cos\alpha_i}{\sin\alpha_j\sin\alpha_k}
,
\end{equation}
and second, using (\ref{eq:29}),
\begin{equation}
\begin{split}
\frac{(\textbf{u}_{ij}\times\textbf{u}_{jk})\cdot(\textbf{u}_{jk}\times\textbf{u}_{ik})}
{|\textbf{u}_{ij}\times\textbf{u}_{jk}||\textbf{u}_{jk}\times\textbf{u}_{ik}|}&=\frac{1}{\sin\alpha_j\sin\alpha_k}\left(\frac{\textbf{n}_i\cdot(\textbf{n}_j\times\textbf{n}_k)}{\sin\theta_{ij}\sin\theta_{jk}}\textbf{n}_j\right)\cdot \left(\frac{\textbf{n}_i\cdot(\textbf{n}_k\times\textbf{n}_j)}{\sin\theta_{ik}\sin\theta_{jk}}\textbf{n}_k\right),\\
&=-\frac{\sin^2(\theta_{ij},\theta_{ik},\theta_{jk})}{\sin\alpha_j\sin\alpha_k\sin\theta_{ij}\sin^2\theta_{jk}\sin\theta_{ik}}\cos\theta_{jk}.
\end{split}
\end{equation}
Reducing this using the Sine Rule and equating with the previous result gives the Polar Cosine Rule.
\end{proof}

Note that, by now considering
\begin{equation}
k^2=\frac{\sin^2\alpha_i}{\sin^2\theta_{jk}}
=\frac{\sin^2\alpha_i}{1-\cos^2\theta_{jk}},
\end{equation}
and substituting in the polar cosine rule,
\begin{equation}
k^2=\frac{\sin^2\alpha_i}{1-\left(\frac{\cos\alpha_i+\cos\alpha_j\cos\alpha_k}{\sin\alpha_j\sin\alpha_k}\right)^2}
=\frac{\sin^2\alpha_i\sin^2\alpha_j\sin^2\alpha_k}{1-\cos^2\alpha_i-\cos^2\alpha_j-\cos^2\alpha_k-2\cos\alpha_i\cos\alpha_j\cos\alpha_k},
\end{equation}
the constant $k$ may also be written in terms of the spherical angles, so that the sine rule now becomes
\begin{equation}\label{eq:ssine}
\frac{\sin\alpha_i}{\sin\theta_{jk}}=\frac{\sin\alpha_j}{\sin\theta_{ik}}=\frac{\sin\alpha_k}{\sin\theta_{ij}}=\frac{\sin\alpha_i\sin\alpha_j\sin\alpha_k}{\sqrt{1-\cos^2\alpha_i-\cos^2\alpha_j-\cos^2\alpha_k-2\cos\alpha_i\cos\alpha_j\cos\alpha_k}}.
\end{equation}
We will use both forms of the sine rule laterwhen discussing the link between spherical trigonometry 
and the Jacobi elliptic functions.

\subsection{Hyperspherical Trigonometry in 4-Dimensional Euclidean Space}

We now extend these principles to the four dimensional hyperspherical case. Consider a $3$-sphere embedded in four dimensional Euclidean space, $\mathbf{E}_4=\mathbb{R}^4$. In this case, we have four unit vectors, $\textbf{n}_i, i=1,\dots,4$ pointing to the four vertices of a hyperspherical tetrahedron. The angles between these unit vectors, $\theta_{ij}$, are defined, in the same way as for the spherical case, by
\begin{equation}
\textbf{n}_i\cdot\textbf{n}_j\equiv\cos\theta_{ij},
\end{equation}
whereas now we must also define orthogonal vectors $\textbf{u}_{ijk}$ to each hyperplane $[i,j,k]$ (defined by $\mathbf{n}_i$, $\mathbf{n}_j$, $\mathbf{n}_k$). This is done using the ternary cross product,
\begin{equation}
\textbf{u}_{ijk}\equiv\frac
{\textbf{n}_i\times\textbf{n}_j\times\textbf{n}_k}
{\left|\textbf{n}_i\times\textbf{n}_j\times\textbf{n}_k\right|}.
\end{equation}

\begin{figure}[hpb]
\centering
\tdplotsetmaincoords{30}{30}
\begin{tikzpicture}[scale=6,tdplot_main_coords]
\draw[thick](0,0,0.8) -- (0,0,3.1);
\draw[thick](0,0,2) -- (0,0,2.3)node[anchor=west]{$\theta_{ik}$};
\tdplotdrawarc[thick]{(0,0,0)}{1}{85}{104}{}{}
\tdplotdrawarc[thick]{(0,0,0)}{1}{136}{155}{}{}
\tdplotdrawarc[dashed]{(0,0,0)}{1}{104}{136}{}{}
\tdplotdrawarc{(0,0,1)}{0.1}{120}{90}{anchor=south west}{$\alpha_k^{(ijk)}$}
\fill[black] (0,1) circle (0.5pt)node[anchor=south west]{$\mathbf{n}_j$};

\draw[thick](0,0,2) -- (-0.1,0,2);
\tdplotsetrotatedcoords{0}{0}{-30}
\draw[tdplot_rotated_coords,thick](0,0,2) -- (0,0.1,2);
\tdplotdrawarc{(0,0,2)}{0.05}{180}{60}{anchor=south}{$\phi_{ik}$}

\tdplotsetrotatedcoords{60}{0}{0}
\fill[tdplot_rotated_coords,black] (0,1) circle (0.5pt)node[anchor=south east]{$\mathbf{n}_l$};

\tdplotsetrotatedcoords{60}{90}{90}
\tdplotdrawarc[tdplot_rotated_coords,thick]{(0,0,0)}{1}{100}{-10}{anchor=north}{$\theta_{kl}$}
\fill[tdplot_rotated_coords,black] (0,1) circle (0.5pt);

\tdplotsetrotatedcoords{0}{90}{90}
\tdplotdrawarc[tdplot_rotated_coords,thick]{(0,0,0)}{1}{100}{-10}{anchor=north}{$\theta_{jk}$}
\fill[tdplot_rotated_coords,black] (0,1) circle (0.5pt)node[anchor=west]{$\mathbf{n}_k$};
\fill[tdplot_rotated_coords,black] (0,2.9) circle (0.5pt)node[anchor=west]{$\mathbf{n}_i$};

\tdplotsetrotatedcoords{90}{90}{0}
\coordinate (Shift) at (-1,1,0);
\tdplotsetrotatedcoordsorigin{(Shift)}
\tdplotdrawarc[tdplot_rotated_coords,thick]{(0,0,0)}{1}{200}{280}{anchor=south west}{$\theta_{ij}$}

\tdplotsetrotatedcoords{150}{90}{0}
\coordinate (Shift) at (0,1,1);
\tdplotsetrotatedcoordsorigin{(Shift)}
\tdplotdrawarc[tdplot_rotated_coords,thick]{(0,0.15,0)}{1}{80}{160}{anchor=south east}{$\theta_{il}$}

\tdplotsetrotatedcoords{0}{0}{180}
\tdplotdrawarc{(0,0,2.95)}{0.1}{335}{300}{anchor=north}{$\alpha_i^{(ijk)}$}

\end{tikzpicture}
\caption{A hyperspherical tetrahedron}
\label{fig:hyperspherical_tetrahedron}
\end{figure}
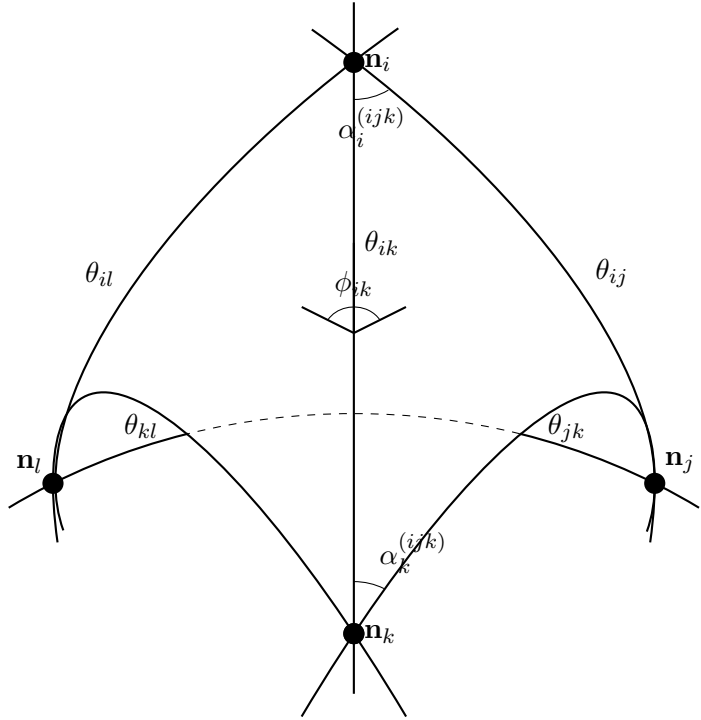

The four vectors $\mathbf{u}_{123}$, $\mathbf{u}_{124}$, $\mathbf{u}_{134}$, $\mathbf{u}_{234}$ define the polar of the hyperspherical tetrahedron, as illustrated in Figure \ref{fig:hyperspherical_tetrahedron}, between which we define dihedral angles $\phi_{jk}$ by
\begin{equation}
\textbf{u}_{ijk}\cdot\textbf{u}_{jkl}\equiv-\cos\phi_{jk}.
\end{equation}
We also have 
\begin{equation}
|\textbf{n}_i\times\textbf{n}_j\times\textbf{n}_k|=\sin(\theta_{ij},\theta_{jk},\theta_{ik}),
\end{equation}
together with
\begin{equation}
|\textbf{u}_{ijk}\times\textbf{u}_{jkl}\times\textbf{u}_{kli}|=\sin(\phi_{jk},\phi_{ik},\phi_{kl}).
\end{equation}

The four faces of the hyperspherical tetrahedron 
\ref{fig:hyperspherical_tetrahedron} 
are spherical triangles, with the various sine and cosine rules for the spherical case still holding true for these. For ease of notation, we now label the angles in the faces of the spherical triangles $\alpha_i^{(ijk)}$ to indicate which triangle is being considered.

In addition to the spherical relations, there are now various relations involving the dihedral angles, $\phi_{ij}$, which connect the various faces of the hyperspherical tetrahedron. By taking various relations between the scalar and vector products for the vectors $\textbf{u}_{ijk}$, we can again derive  a number of relations.
\begin{prop}[Cosine Rule]
\begin{equation}\label{eq:cosine}
\cos\phi_{jk}=\frac{-\left|\begin{array}{ccc}
\cos\theta_{ij} & \cos\theta_{ik} & \cos\theta_{il} \\
1 & \cos\theta_{jk} & \cos\theta_{jl} \\
\cos\theta_{jk} & 1 & \cos\theta_{kl} \\
\end{array}\right|}
{\sin(\theta_{ij},\theta_{ik},\theta_{jk})\sin(\theta_{jk},\theta_{jl},\theta_{kl})}.
\end{equation} for all $i,j,k,l=1,2,3,4$. 
\end{prop}

\begin{proof}
From the computation
\begin{equation}
\begin{split}
(\textbf{n}_i\times\textbf{n}_j\times\textbf{n}_k)\cdot(\textbf{n}_j\times\textbf{n}_k\times\textbf{n}_l)&=-(\mathbf{n}_j\times\mathbf{n}_k\times(\mathbf{n}_j\times\mathbf{n}_k\times\mathbf{n}_l))\mathbf{n}_i,\\
&=\left|\begin{array}{ccc}
\textbf{n}_j & \textbf{n}_k & \textbf{n}_l \\
\textbf{n}_j\cdot\textbf{n}_j & \textbf{n}_j\cdot\textbf{n}_k & \textbf{n}_j\cdot\textbf{n}_l \\
\textbf{n}_k\cdot\textbf{n}_j & \textbf{n}_k\cdot\textbf{n}_k & \textbf{n}_k\cdot\textbf{n}_l \\
\end{array}\right|\cdot\mathbf{n}_i,\\
&=\left|\begin{array}{ccc}
\textbf{n}_i\cdot\textbf{n}_j & \textbf{n}_i\cdot\textbf{n}_k & \textbf{n}_i\cdot\textbf{n}_l \\
\textbf{n}_j\cdot\textbf{n}_j & \textbf{n}_j\cdot\textbf{n}_k & \textbf{n}_j\cdot\textbf{n}_l \\
\textbf{n}_k\cdot\textbf{n}_j & \textbf{n}_k\cdot\textbf{n}_k & \textbf{n}_k\cdot\textbf{n}_l \\
\end{array}\right|,
\end{split}
\end{equation}
we find that for the following inner product (for all $i,j,k,l=1,2,3,4$) we have
\begin{equation}
\begin{split}
\textbf{u}_{ijk}\cdot\textbf{u}_{jkl}=&\frac
{(\textbf{n}_i\times\textbf{n}_j\times\textbf{n}_k)\cdot(\textbf{n}_j\times\textbf{n}_k\times\textbf{n}_l)}
{\left|\textbf{n}_i\times\textbf{n}_j\times\textbf{n}_k\right|\left|\textbf{n}_j\times\textbf{n}_k\times\textbf{n}_l\right|},\\
=&\frac{\left|\begin{array}{ccc}
\textbf{n}_i\cdot\textbf{n}_j & \textbf{n}_i\cdot\textbf{n}_k & \textbf{n}_i\cdot\textbf{n}_l \\
\textbf{n}_j\cdot\textbf{n}_j & \textbf{n}_j\cdot\textbf{n}_k & \textbf{n}_j\cdot\textbf{n}_l \\
\textbf{n}_k\cdot\textbf{n}_j & \textbf{n}_k\cdot\textbf{n}_k & \textbf{n}_k\cdot\textbf{n}_l \\
\end{array}\right|}
{\sin(\theta_{ij},\theta_{ik},\theta_{jk})\sin(\theta_{jk},\theta_{jl},\theta_{kl})}.\\
\end{split}
\end{equation}
Hence, the Cosine Rule follows.
\end{proof}
We now need to extend to definition of the triple sine function for the four-dimensional case, to give a further generalisation of the sine function dependent on six variables.
\begin{defn}
The \textbf{generalised} \textbf{sine} \textbf{function} \textbf{of} \textbf{six} \textbf{variables}, $\sin(\theta_{ij},\theta_{ik},\theta_{il},\theta_{jk},\theta_{jl},\theta_{kl})$, is the four dimensional analogue of the triple sine function, and is given by
\begin{equation}
\sin(\theta_{ij},\theta_{ik},\theta_{il},\theta_{jk},\theta_{jl},\theta_{kl})=\left|\begin{array}{cccc}
1&
\cos\theta_{ij}&
\cos\theta_{ik}&
\cos\theta_{il}\\
\cos\theta_{ij}&
1&
\cos\theta_{jk}&
\cos\theta_{jl}\\
\cos\theta_{ik}&
\cos\theta_{jk}&
1&
\cos\theta_{kl}\\
\cos\theta_{il}&
\cos\theta_{jl}&
\cos\theta_{kl}&
1\\
\end{array}\right|^{\frac{1}{2}},\end{equation}
\begin{equation}
0\leq\theta_{ij},\theta_{ik},\theta_{il},\theta_{jk},\theta_{jl},\theta_{kl}\leq\pi.
\end{equation}
Note that the restrictions on being a spherical tetrahedron ensure that the generalised sine function is always real and positive.
\end{defn}
This sine function is equivalent to the volume of a 4D parallelotope with edges $\mathbf{n}_i$, $\mathbf{n}_j$, $\mathbf{n}_k$ and $\mathbf{n}_l$, given by
\begin{equation}
\text{Volume}=|\textbf{n}_i\cdot(\textbf{n}_j\times\textbf{n}_k\times\textbf{n}_l)|=\sin(\theta_{ij},\theta_{ik},\theta_{il},\theta_{jk},\theta_{jl},\theta_{kl}).
\end{equation}
This relation follows from embedding 4D vectors $\mathbf{n}_i$, $\mathbf{n}_j$, $\mathbf{n}_k$ and $\mathbf{n}_l$ into $E_5$ in a similar manner to the three-dimensional case.

\begin{prop}[Sine Rule]
\begin{equation}\label{eq:sin}
\begin{split}
\frac{\sin(\phi_{ij},\phi_{ik},\phi_{il})}{\sin(\theta_{jk},\theta_{jl},\theta_{kl})}
&=\frac{\sin(\phi_{ij},\phi_{jk},\phi_{jl})}{\sin(\theta_{ik},\theta_{il},\theta_{kl})}
=\frac{\sin(\phi_{ik},\phi_{jk},\phi_{kl})}{\sin(\theta_{ij},\theta_{il},\theta_{jl})}
=\frac{\sin(\phi_{il},\phi_{jl},\phi_{kl})}{\sin(\theta_{ij},\theta_{ik},\theta_{jk})},\\
&=\frac{\sin^2(\theta_{ij},\theta_{ik},\theta_{il},\theta_{jk},\theta_{jl},\theta_{kl})}{\sin(\theta_{ij},\theta_{ik},\theta_{jk})\sin(\theta_{ij},\theta_{il},\theta_{jl})\sin(\theta_{ik},\theta_{il},\theta_{kl})\sin(\theta_{jk},\theta_{jl},\theta_{kl})}=k_H,
\end{split}
\end{equation}
where $k_H$ is constant.
\end{prop}

\begin{proof}
Consider the triple product
\begin{equation}\label{eq:47}
\begin{split}
\textbf{u}_{ijk}\times\textbf{u}_{jkl}\times\textbf{u}_{kli}&=\frac{(\textbf{n}_i\times\textbf{n}_j\times\textbf{n}_k)\times(\textbf{n}_j\times\textbf{n}_k\times\textbf{n}_l)\times(\textbf{n}_k\times\textbf{n}_l\times\textbf{n}_i)}
{|\textbf{n}_i\times\textbf{n}_j\times\textbf{n}_k||\textbf{n}_j\times\textbf{n}_k\times\textbf{n}_l||\textbf{n}_k\times\textbf{n}_l\times\textbf{n}_i|},\\
&=\frac{-\left|\begin{array}{ccc}
\textbf{n}_i & \textbf{n}_j & \textbf{n}_k \\
\textbf{n}_i\cdot(\textbf{n}_j\times\textbf{n}_k\times\textbf{n}_l) & 0& 0\\
0 & \textbf{n}_j\cdot(\textbf{n}_k\times\textbf{n}_l\times\textbf{n}_i)& 0
\end{array}\right|}
{\sin(\theta_{ij},\theta_{ik},\theta_{jk})\sin(\theta_{jk},\theta_{jl},\theta_{kl})\sin(\theta_{kl},\theta_{ki},\theta_{li})},\\
&=\frac{\left(\textbf{n}_i\cdot(\textbf{n}_j\times\textbf{n}_k\times\textbf{n}_l)\right)^2}{\sin(\theta_{ij},\theta_{ik},\theta_{jk})\sin(\theta_{jk},\theta_{jl},\theta_{kl})\sin(\theta_{kl},\theta_{ki},\theta_{li})}\,\textbf{n}_k.
\end{split}
\end{equation}
This implies
\begin{equation}
\begin{split}
|\textbf{u}_{ijk}\times\textbf{u}_{jkl}\times\textbf{u}_{kli}|&=\frac{\left(\textbf{n}_i\cdot(\textbf{n}_j\times\textbf{n}_k\times\textbf{n}_l)\right)^2}{\sin(\theta_{ij},\theta_{ik},\theta_{jk})\sin(\theta_{jk},\theta_{jl},\theta_{kl})\sin(\theta_{kl},\theta_{ki},\theta_{li})},\\
&=\frac{\sin^2(\theta_{ij},\theta_{ik},\theta_{il},\theta_{jk},\theta_{jl},\theta_{kl})}{\sin(\theta_{ij},\theta_{ik},\theta_{jk})\sin(\theta_{jk},\theta_{jl},\theta_{kl})\sin(\theta_{kl},\theta_{ki},\theta_{li})}.
\end{split}
\end{equation}
From this we have that the ratio
\begin{equation}\label{ratio2}
\frac{|\textbf{u}_{ijk}\times\textbf{u}_{jkl}\times\textbf{u}_{kli}|}{|\textbf{n}_i\times\textbf{n}_j\times\textbf{n}_l|}=\frac{\sin(\phi_{ik},\phi_{jk},\phi_{kl})}{\sin(\theta_{ij},\theta_{il},\theta_{jl})}
\end{equation}
is symmetric with respect to the interchange of the labels $i,j,k,l$. Thus, it follows that
\begin{equation}
\frac{\sin(\phi_{ik},\phi_{jk},\phi_{kl})}{\sin(\theta_{ij},\theta_{il},\theta_{jl})}=\frac{\sin^2(\theta_{ij},\theta_{ik},\theta_{il},\theta_{jk},\theta_{jl},\theta_{kl})}{\sin(\theta_{ij},\theta_{ik},\theta_{jk})\sin(\theta_{jk},\theta_{jl},\theta_{kl})\sin(\theta_{kl},\theta_{ki},\theta_{li})\sin(\theta_{ij},\theta_{il},\theta_{jl})},
\end{equation}
and hence, the hyperspherical sine rule follows.
\end{proof}

There also exists a simpler sine relationship between the standard sine functions of the central angles and the sine of the dihedral angles.
\begin{prop}
\begin{equation}\label{eq:sinn}
\sin\phi_{kl}=\frac{\sin(\theta_{ij},\theta_{ik},\theta_{il},\theta_{jk},\theta_{jl},\theta_{kl})}{\sin(\theta_{ik},\theta_{il},\theta_{kl})\sin(\theta_{jk},\theta_{jl},\theta_{kl})}\sin\theta_{kl}.
\end{equation}
\end{prop}
In order to prove this proposition, we require a determinantal identity, attributed to Desnanot for the  $n\leq 6$ case\cite{Desnanot}, and to Jacobi in the general case.\cite{Jacobi2}
\begin{theorem}[Desnanot-Jacobi Identity]
Let $M$ be a $n\times n$ square matrix, and denote by $M^p_i$ the matrix obtained by removing both the $i$-th row and $p$-th column. Similarly, let $M^{p,q}_{i,j}$ denote the matrix obtained by deleting the $i$-th and $j$-th rows, and the $p$-th and $q$-th columns, respectively, with $1\leq i,j,p,q\leq n$. Then,
\begin{equation}
\mathrm{det}\left(M\right)\mathrm{det}\left(M^{1,n}_{1,n}\right)=\mathrm{det}\left(M^{1}_{1}\right)\mathrm{det}\left(M^{n}_{n}\right)-\mathrm{det}\left(M^{n}_{1}\right)\mathrm{det}\left(M^{1}_{n}\right).
\end{equation}
Diagrammatically, this is
\begin{equation}\label{eq:Sylv}
\left|
\begin{picture}(49,32)
\end{picture}
\right|\times
\left|
\begin{color}{black}
\begin{picture}(49,32)
\put(3,-22){\rule{.2mm}{1.8cm}}\put(46,-22){\rule{.2mm}{1.8cm}}
\put(1,-19){\rule{1.7cm}{.2mm}}\put(1,26){\rule{1.7cm}{.2mm}}
\end{picture}
\end{color}
\right| =
\left|
\begin{color}{black}
\begin{picture}(49,32)
\put(3,-22){\rule{.2mm}{1.8cm}}\put(1,26){\rule{1.7cm}{.2mm}}
\end{picture}
\end{color}
\right|\times\left|
\begin{color}{black}
\begin{picture}(49,32)
\put(46,-22){\rule{.2mm}{1.8cm}}\put(1,-19){\rule{1.7cm}{.2mm}}
\end{picture}
\end{color}
\right|
-\left|
\begin{color}{black}
\begin{picture}(49,32)
\put(3,-22){\rule{.2mm}{1.8cm}}\put(1,-19){\rule{1.7cm}{.2mm}}
\end{picture}
\end{color}
\right|\times\left|
\begin{color}{black}
\begin{picture}(49,32)
\put(46,-22){\rule{.2mm}{1.8cm}}\put(1,26){\rule{1.7cm}{.2mm}}
\end{picture}
\end{color}
\right|.
\end{equation}
\end{theorem}

\begin{proof}
The result follows directly from the Desnanot-Jacobi determinant identity
\begin{equation}
\begin{split}
&\left|\begin{array}{cccc}
1& \cos\theta_{ij} & \cos\theta_{ik} & \cos\theta_{jk}\\
\cos\theta_{ij} & 1 & \cos\theta_{jk} & \cos\theta_{jl}\\
\cos\theta_{ik} & \cos\theta_{jk} & 1 & \cos\theta_{kl}\\
\cos\theta_{il} & \cos\theta_{jl} & \cos\theta_{kl} & 1\\
\end{array}\right|
\left|\begin{array}{cc}
 1 & \cos\theta_{kl}\\
 \cos\theta_{kl} & 1\\
\end{array}\right|\\
=&
\left|\begin{array}{ccc}
1 & \cos\theta_{jk} & \cos\theta_{jl}\\
\cos\theta_{jk} & 1 & \cos\theta_{kl}\\
\cos\theta_{jl} & \cos\theta_{kl} & 1\\
\end{array}\right|
\left|\begin{array}{ccc}
1 & \cos\theta_{ik} & \cos\theta_{jk}\\
\cos\theta_{ik} &  1 & \cos\theta_{kl}\\
\cos\theta_{il} & \cos\theta_{kl} & 1\\
\end{array}\right|\\
&-\left|\begin{array}{ccc}
\cos\theta_{ij} &  \cos\theta_{jk} & \cos\theta_{jl}\\
\cos\theta_{ik}  & 1 & \cos\theta_{kl}\\
\cos\theta_{il}  & \cos\theta_{kl} & 1\\
\end{array}\right|\left|\begin{array}{ccc}
\cos\theta_{ij} & \cos\theta_{ik} & \cos\theta_{jk}\\
\cos\theta_{jk} & 1 & \cos\theta_{kl}\\
\cos\theta_{jl} & \cos\theta_{kl} & 1\\
\end{array}\right|.
\end{split}
\end{equation}
The result follows directly. 
\end{proof}

From this proposition, it hence follows that the hyperspherical sine rule may be rewritten as 
\begin{equation}
\frac{\sin\phi_{ij}}{\sin\theta_{ij}}\frac{\sin\phi_{kl}}{\sin\theta_{kl}}
=\frac{\sin\phi_{ik}}{\sin\theta_{ik}}\frac{\sin\phi_{jl}}{\sin\theta_{jl}}
=\frac{\sin\phi_{il}}{\sin\theta_{il}}\frac{\sin\phi_{jk}}{\sin\theta_{jk}}
=k_H.
\end{equation}
The constant $k_H$ may also be neatly expressed in terms of cosines.\cite{Derevnin} For this we need the Polar Cosine rule.

\begin{prop}[Polar Cosine Rule]
\begin{equation}
\cos\theta_{kl}=\frac{\left|\begin{array}{ccc}
-\cos\phi_{jk} &
\cos\phi_{ik} &
-\cos\phi_{ij} \\
1 &
-\cos\phi_{kl} &
\cos\phi_{jl} \\
-\cos\phi_{kl} &
1 &
-\cos\phi_{il} 
\end{array}\right|}
{\sin(\phi_{ik},\phi_{jk},\phi_{kl})\sin(\phi_{il},\phi_{jl},\phi_{kl})}.
\end{equation}for all $i,j,k,l=1,2,3,4$.
\end{prop}

\begin{proof}
Consider the product
\begin{equation}
\begin{split}
\frac{(\textbf{u}_{ijk}\times\textbf{u}_{jkl}\times\textbf{u}_{kli})\cdot(\textbf{u}_{jkl}\times\textbf{u}_{kli}\times\textbf{u}_{lij})}
{|\textbf{u}_{ijk}\times\textbf{u}_{jkl}\times\textbf{u}_{kli}||\textbf{u}_{jkl}\times\textbf{u}_{kli}\times\textbf{u}_{lij}|}&=
\frac{\left|\begin{array}{ccc}
\textbf{u}_{ijk}\cdot\textbf{u}_{jkl} &
\textbf{u}_{ijk}\cdot\textbf{u}_{kli} &
\textbf{u}_{ijk}\cdot\textbf{u}_{lij} \\
\textbf{u}_{jkl}\cdot\textbf{u}_{jkl} &
\textbf{u}_{jkl}\cdot\textbf{u}_{kli} &
\textbf{u}_{jkl}\cdot\textbf{u}_{lij} \\
\textbf{u}_{kli}\cdot\textbf{u}_{jkl} &
\textbf{u}_{kli}\cdot\textbf{u}_{kli} &
\textbf{u}_{kli}\cdot\textbf{u}_{lij} 
\end{array}\right|}
{\sin(\phi_{ik},\phi_{jk},\phi_{kl})\sin(\phi_{il},\phi_{jl},\phi_{kl})},\\
&=
\frac{\left|\begin{array}{ccc}
-\cos\phi_{jk} &
\cos\phi_{ik} &
-\cos\phi_{ij} \\
1 &
-\cos\phi_{kl} &
\cos\phi_{jl} \\
-\cos\phi_{kl} &
1 &
-\cos\phi_{il} 
\end{array}\right|}
{\sin(\phi_{ik},\phi_{jk},\phi_{kl})\sin(\phi_{il},\phi_{jl},\phi_{kl})}.
\end{split}
\end{equation}
On the other hand, using (\ref{eq:47}), we have
\begin{equation}
\textbf{u}_{ijk}\times\textbf{u}_{jkl}\times\textbf{u}_{kli}=\frac{\left(\textbf{n}_i\cdot(\textbf{n}_j\times\textbf{n}_k\times\textbf{n}_l)\right)^2}{\sin(\theta_{ij},\theta_{ik},\theta_{jk})\sin(\theta_{jk},\theta_{jl},\theta_{kl})\sin(\theta_{kl},\theta_{ki},\theta_{li})}\,\textbf{n}_k.
\end{equation}
Reducing this using the Sine Rule and equating the two results gives the hyperspherical polar cosine rule.
\end{proof}

\begin{prop}
\begin{equation}
\begin{split}
\frac{\cos\phi_{ij}\cos\phi_{kl}-\cos\phi_{ik}\cos\phi_{jl}}{\cos\theta_{ij}\cos\theta_{kl}-\cos\theta_{ik}\cos\theta_{jl}}&=\frac{\cos\phi_{ik}\cos\phi_{jl}-\cos\phi_{il}\cos\phi_{jk}}{\cos\theta_{ik}\cos\theta_{jl}-\cos\theta_{il}\cos\theta_{jk}},\\
&=\frac{\cos\phi_{il}\cos\phi_{jk}-\cos\phi_{ij}\cos\phi_{kl}}{\cos\theta_{il}\cos\theta_{jk}-\cos\theta_{ij}\cos\theta_{kl}},\\
&=k_H.
\end{split}
\end{equation}
\end{prop}

\begin{proof}
Consider the numerator
\begin{equation}
\cos\phi_{ij}\cos\phi_{kl}-\cos\phi_{ik}\cos\phi_{jl},
\end{equation}
and rewrite this in terms of the central angles using the cosine rule,
\begin{equation}
\begin{split}
\cos\phi_{ij}\cos\phi_{kl}&-\cos\phi_{ik}\cos\phi_{jl}=
\frac{\left|\begin{array}{ccc}
\cos\theta_{ik} & \cos\theta_{jk} & \cos\theta_{kl}\\
1 & \cos\theta_{ij} & \cos\theta_{il}\\
\cos\theta_{ij} & 1 & \cos\theta_{jl}\\
\end{array}\right|}
{\sin(\theta_{ij},\theta_{ik},\theta_{jk})\sin(\theta_{ij},\theta_{il},\theta_{jl})}
\frac{\left|\begin{array}{ccc}
\cos\theta_{ik} & \cos\theta_{il} & \cos\theta_{ij}\\
1 & \cos\theta_{kl} & \cos\theta_{jk}\\
\cos\theta_{kl} & 1 & \cos\theta_{jl}\\
\end{array}\right|}
{\sin(\theta_{ik},\theta_{il},\theta_{kl})\sin(\theta_{jk},\theta_{jl},\theta_{kl})}\\
&-
\frac{\left|\begin{array}{ccc}
\cos\theta_{ij} & \cos\theta_{jk} & \cos\theta_{jl}\\
1 & \cos\theta_{ik} & \cos\theta_{il}\\
\cos\theta_{ik} & 1 & \cos\theta_{kl}\\
\end{array}\right|}
{\sin(\theta_{ij},\theta_{ik},\theta_{jk})\sin(\theta_{ik},\theta_{il},\theta_{jk})}
\frac{\left|\begin{array}{ccc}
\cos\theta_{ij} & \cos\theta_{il} & \cos\theta_{ik}\\
1 & \cos\theta_{jl} & \cos\theta_{jk}\\
\cos\theta_{jl} & 1 & \cos\theta_{kl}\\
\end{array}\right|}
{\sin(\theta_{ij},\theta_{il},\theta_{jl})\sin(\theta_{jk},\theta_{jl},\theta_{kl})}.
\end{split}
\end{equation}
Multiplying this out and factorising, this reduces to
\begin{equation}
\cos\phi_{ij}\cos\phi_{kl}-\cos\phi_{ik}\cos\phi_{jl}=(\cos\theta_{ij}\cos\theta_{kl}-\cos\theta_{ik}\cos\theta_{jl})k_H,
\end{equation}
from which the result follows.
\end{proof}
Furthermore, in 4D there are also relations involving the spherical angles $\alpha_j^{(j\cdot \cdot)}$ at a common vertex, $\mathbf{n}_j$, of the tetrahedron involving the different spherical triangles.
\begin{prop}[Vertex Cosine Rule]
\begin{equation}
\cos\phi_{jk}=\frac{\cos\alpha_j^{(ijk)}\cos\alpha_j^{(jkl)}-\cos\alpha_j^{(ijl)}}{\sin\alpha_j^{(ijk)}\sin\alpha_j^{(jkl)}}.
\end{equation}
\end{prop}
\begin{proof}
Expanding out the cosine rule (\ref{eq:cosine}) and applying the sine rule (\ref{eq:sin}) to the denominator gives
\begin{equation}
\begin{split}
\cos\phi_{jk}=&(\cos\theta_{ij}\cos\theta_{jl}+\cos\theta_{ik}\cos\theta_{kl}-\cos\theta_{ij}\cos\theta_{jk}\cos\theta_{kl}\\
&-\cos\theta_{ik}\cos\theta_{jk}\cos\theta_{jl}-\cos\theta_{il}+\cos\theta_{il}\cos^2\theta_{jk})\\
&/{\sin\alpha_j^{(ijk)}\sin\alpha_j^{(jkl)}\sin\theta_{ij}\sin^2\theta_{jk}\sin\theta_{jl}}.
\end{split}
\end{equation}
This can be factorised as
\begin{equation}
\cos\phi_{jk}=\frac{\left(\frac{\cos\theta_{ik}-\cos\theta_{ij}\cos\theta_{jk}
}{\sin\theta_{ij}\sin\theta_{jk}}\right)\left(\frac{\cos\theta_{kl}-\cos\theta_{jk}\cos\theta_{jl}
}{\sin\theta_{jk}\sin\theta_{jl}}\right)-\left(\frac{\cos\theta_{il}-\cos\theta_{ij}\cos\theta_{jl}
}{\sin\theta_{ij}\sin\theta_{jl}}\right)}
{\sin\alpha_j^{(ijk)}\sin\alpha_j^{(jkl)}}.
\end{equation}
Applying (\ref{eq:23}) the result follows.
\end{proof}

This set of relations imply, in turn, corresponding vertex polar cosine relations, expressing the spherical angles at a vertex in terms of the dihedral angles as follows,
\begin{equation}
\cos\alpha_j^{(ijl)}=\frac{\cos\phi_{jk}+\cos\phi_{ij}\cos\phi_{jl}}{\sin\phi_{ij}\sin\phi_{jl}},
\end{equation}
together with a corresponding sine rule of the form
\begin{equation}
\frac{\sin\phi_{jk}}{\sin\alpha_{j}^{(ijl)}}=\frac{\sin\phi_{jl}}{\sin\alpha_{j}^{(ijk)}}=\frac{\sin\phi_{ij}}{\sin\alpha_{j}^{(jkl)}}=\frac{\sin(\alpha_{j}^{(ijk)},\alpha_{j}^{(ijl)},\alpha_{j}^{(jkl)})}{\sin\alpha_{j}^{(ijk)}\sin\alpha_{j}^{(ijl)}\sin\alpha_{j}^{(jkl)}},
\end{equation}
for $i,j,k,l=1,2,3,4$. Note the similarities between these formulae and those of the spherical case.

%%% Extra from 15/3 2015 version 
We also present a, to our knowledge, new formula which again follows from the Desnanot-Jacobi 
identity, which proves particularly useful when providing the link between hyperspherical trigonometry 
and elliptic functions given later.  
\begin{prop} The following identity holds between the various tetrahedral angles: 
\begin{equation}
\begin{split} 
&  \frac{\sin^2(\theta_{ij},\theta_{ik},\theta_{il},\theta_{jk},\theta_{jl},\theta_{kl})\, 
\cos\alpha_l^{(jkl)}\,\sin\theta_{jl}\,\sin\theta_{kl}}{\sin^2(\theta_{jk},\theta_{jl},\theta_{kl})} \\ 
& = \sin(\theta_{ij},\theta_{il},\theta_{jl})\,\sin(\theta_{ik},\theta_{il},\theta_{kl})\,
\left(\cos\phi_{jl}\,\cos\phi_{kl}-\cos\phi_{il}\right). 
\end{split}
\end{equation}
\end{prop} 
\begin{proof}
Consider the Desnanot-Jacobi determinantal identity applied to 
$\sin^4(\theta_{ij},\theta_{ik},\theta_{il},\theta_{jk},\theta_{jl},\theta_{kl})$. This yields: 
\begin{equation}
\begin{split} 
\left| \begin{array}{cc} \cos\theta_{jk} & \cos\theta_{jl} \\ \cos\theta_{kl} & 1\end{array}\right|&
\left| \begin{array}{cccc} 1 & \cos\theta_{ij} & \cos\theta_{ik} & \cos\theta_{il} \\ 
\cos\theta_{ij} & 1 & \cos\theta_{jk} & \cos\theta_{jl} \\ 
\cos\theta_{ik} & \cos\theta_{jk} & 1 & \cos\theta_{kl} \\
\cos\theta_{il} & \cos\theta_{jl}& \cos\theta_{kl} & 1 \end{array}\right| \\ 
&= \left| \begin{array}{ccc}  
 1 & \cos\theta_{jk} & \cos\theta_{jl} \\ 
 \cos\theta_{jk} & 1 & \cos\theta_{kl} \\
 \cos\theta_{jl}& \cos\theta_{kl} & 1 \end{array}\right| \, 
\left| \begin{array}{ccc} 1  & \cos\theta_{ik} & \cos\theta_{il} \\ 
\cos\theta_{ij} & \cos\theta_{jk} & \cos\theta_{jl} \\ 
\cos\theta_{il} & \cos\theta_{kl} & 1 \end{array}\right| \\ 
& - \left| \begin{array}{ccc}  \cos\theta_{ij} & \cos\theta_{ik} & \cos\theta_{il} \\ 
 1 & \cos\theta_{jk} & \cos\theta_{jl} \\ 
 \cos\theta_{jl}& \cos\theta_{kl} & 1 \end{array}\right|\, 
\left| \begin{array}{ccc} 
\cos\theta_{ij}  & \cos\theta_{jk} & \cos\theta_{jl} \\ 
\cos\theta_{ik}  & 1 & \cos\theta_{kl} \\
\cos\theta_{il} & \cos\theta_{kl} & 1 \end{array}\right|
\end{split} 
\end{equation}

from which the result follows. 
\end{proof}

%%%%%

In Appendix A, we consider the expansion of one of the radial vectors, $\mathbf{n}_i$, in terms of the other radial vectors. We compare this with the expansion in terms of an orthogonal frame, obtained by Gram-Schmidt orthonormalisation to obtain the four-parts formula. We follow a similar procedure for the orthogonal vectors, $\mathbf{u}_{ijk}$, to obtain the five-parts formula. We believe these expansions will be of particular use in considering a model associated with quadrilateral nets introduced by Sergeev in \cite{Sergeev}. 

\subsection{m-Dimensional Hyperspherical Trigonometry}

Some of the results in this subsection were derived in \cite{sato} from a different perspective. 
Consider an $(m-1)$-dimensional hyperspherical simplex on the surface of an $m$-dimensional hypersphere, an $(m-1)$-sphere. Let the $m$ vectors, $\mathbf{n}_1,\mathbf{n}_2,\dots,\mathbf{n}_m$ be the position vectors of the $m$ vertices of this simplex, such that the edges of this simplex are given by $\theta^{[1]}_{i_ji_k}$, with
\begin{equation}
\mathbf{n}_{i_j}\cdot \mathbf{n}_{i_k}\equiv\cos\theta^{[1]}_{i_ji_k},
\end{equation}
for $i_j,i_k\in\{1,2,\dots,m\}$. Now define orthogonal vectors $\mathbf{u}_{i_1i_2\dots i_{m-1}}$ using the $(m-1)$-ary vector product,
\begin{equation}
\mathbf{u}_{i_1i_2\dots i_{m-1}}\equiv\frac{\mathbf{n}_{i_1}\times\mathbf{n}_{i_2}\times\dots\times\mathbf{n}_{i_{m-1}}}{|\mathbf{n}_{i_1}\times\mathbf{n}_{i_2}\times\dots\times\mathbf{n}_{i_{m-1}}|}.
\end{equation}
Define the angles between these orthogonal vectors by
\begin{equation}
\mathbf{u}_{i_1i_2\dots i_{m-1}}\cdot\mathbf{u}_{i_2i_3\dots i_{m}}\equiv-\cos\theta^{[m-1](i_1\dots i_{m})}_{i_2\dots i_{m-1}}.
\end{equation}
Note that the superscript $(i_1\dots i_m)$ denotes the simplex that is being considered. In the highest dimensional case this may be omitted.
We also have
\begin{equation}
|\mathbf{n}_{i_1}\times\mathbf{n}_{i_2}\times\dots\times\mathbf{n}_{i_{m-1}}|=\sin\left(\Theta^{[1]}_{i_1\dotsi_{m-1}}\right),
\end{equation}
where $\sin\left(\Theta^{[1]}_{i_1\dotsi_{m-1}}\right)$ is the generalised sine function of $m(m-1)/2$ variables,
\begin{equation}
\sin\left(\Theta^{[1]}_{i_1\dotsi_{m-1}}\right)=\sin\left(\left\{\theta^{[1]}_{i_ji_k}\bigg|i_j< i_k,\, i_j,i_k=1,\dots,m-1\right\}\right),
\end{equation}
together with 
\begin{equation}
|\mathbf{u}_{i_1\dots i_{m-1}}\times\mathbf{u}_{i_2\dots i_{m}}\times\mathbf{u}_{i_3\dots i_{m}i_1}\times\dots\times\mathbf{u}_{i_{m-1}i_mi_1\dots i_{m-3}}|=\sin\left(\Theta^{[m-1]}_{i_1\dots i_m}\right),
\end{equation}
where similarly $\sin\left(\Theta^{[m-1]}_{i_1\dots i_m}\right)$ is the multiple sine of all of the angles between each pair of orthogonal vectors.
There follow a number of identities, similar to those for the lower dimensional cases.

\begin{prop}[Cosine Rule]
\begin{equation}
\cos\theta^{[m-1]}_{i_2\dots i_{m-1}}=-\frac{\left|\begin{array}{ccc}
\cos\theta^{[1]}_{i_1i_2}&\dots &\cos\theta^{[1]}_{i_1i_m}\\
\vdots & \ddots & \vdots\\
\cos\theta^{[1]}_{i_{m-1}i_2} &\dots &\cos\theta^{[1]}_{i_{m-1}i_m}\\
\end{array}\right|}
{\sin\left(\Theta^{[1]}_{i_1\dots i_{m-1}}\right)\sin\left(\Theta^{[1]}_{i_2\dots i_{m}}\right)}.
\end{equation}
\end{prop}

\begin{proof}
Consider the scalar product
\begin{equation}
\begin{split}
\mathbf{u}_{i_1\dots i_{m-1}}\cdot\mathbf{u}_{i_2\dots i_{m}}&=\frac{(\mathbf{n}_{i_1}\times\dots\times\mathbf{n}_{i_{m-1}})\cdot(\mathbf{n}_{i_2}\times\dots\times\mathbf{n}_{i_{m}})}
{|\mathbf{n}_{i_1}\times\dots\times\mathbf{n}_{i_{m-1}}||\mathbf{n}_{i_2}\times\dots\times\mathbf{n}_{i_{m}}|},\\
&=\frac{\left|\begin{array}{ccc}
\mathbf{n}_{i_1}\cdot\mathbf{n}_{i_2}&\dots&\mathbf{n}_{i_1}\cdot\mathbf{n}_{i_m}\\
\vdots &\ddots & \vdots\\
\mathbf{n}_{i_{m-1}}\cdot\mathbf{n}_{i_2}&\dots&\mathbf{n}_{i_{m-1}}\cdot\mathbf{n}_{i_m}\\
\end{array}\right|}
{\sin\left(\Theta^{[1]}_{i_1\dots i_{m-1}}\right)\sin\left(\Theta^{[1]}_{i_2\dots i_{m}}\right)},\\
&=\frac{\left|\begin{array}{ccc}
\cos\theta^{[1]}_{i_1i_2}&\dots &\cos\theta^{[1]}_{i_1i_m}\\
\vdots & \ddots & \vdots\\
\cos\theta^{[1]}_{i_{m-1}i_2} &\dots &\cos\theta^{[1]}_{i_{m-1}i_m}\\
\end{array}\right|}
{\sin\left(\Theta^{[1]}_{i_1\dots i_{m-1}}\right)\sin\left(\Theta^{[1]}_{i_2\dots i_{m}}\right)}.
\end{split}
\end{equation}
Hence, the cosine rule follows.
\end{proof}

\begin{prop}[Sine Rule]
\begin{equation}\label{eq:81}
\frac{\sin\left(\Theta^{[1](i_{m-1})}_{i_1\dots i_{m}}\right)}{\sin\left(\Theta^{[m-1]}_{i_mi_1\dots i_{m-2}}\right)}==\frac{\sin^{m-2}\left(\Theta^{[1]}_{i_1\dots i_m}\right)}{\sin\left(\Theta^{[m-1]}_{i_1\dots i_{m-1}}\right)\dots\sin\left(\Theta^{[m-1]}_{i_{m-1}i_mi_1\dots i_{m-3}}\right)}=k,\;\textrm{a\;constant},
\end{equation}
where
\begin{equation}
\sin\left(\Theta^{[1](i_{m-1})}_{i_1 \dots i_m}\right)=\sin\left(\left\{\theta^{[1]}_{j_1\dots j_{m-1}}\bigg| i_{m-1}=j_k\, \mathrm{for\,some}\, k\in\{1,\dots,m-1\}\right\}\right).
\end{equation}

\end{prop}

\begin{proof}
Consider the ratio
\begin{equation}
\frac{|\mathbf{u}_{i_1\dots i_{m-1}}\times\mathbf{u}_{i_2\dots i_{m}}\times\dots\times\mathbf{u}_{i_{m-1}i_mi_1\dots i_{m-3}}|}{|\mathbf{u}_{i_mi_1\dots i_{m-2}}|}
=\frac{\sin\left(\Theta^{[1](i_{m-1})}_{i_1\dots i_{m}}\right)}{\sin\left(\Theta^{[m-1]}_{i_mi_1\dots i_{m-2}}\right)}.
\end{equation}
However,
\begin{equation}\label{eq:84}
\begin{split}
&(\mathbf{u}_{i_1\dots i_{m-1}}\times\mathbf{u}_{i_2\dots i_{m}}\times\dots\times\mathbf{u}_{i_{m-1}i_mi_1\dots i_{m-3}})\\
&=\frac{(\mathbf{n}_{i_1}\times\dots\times\mathbf{n}_{i_{m-1}})\times\dots\times(\mathbf{n}_{i_{m-1}}\times \mathbf{n}_{i_m}\times \mathbf{n}_{i_{m-1}}\dots\times\mathbf{n}_{i_{m-3}})}{|\mathbf{n}_{i_1}\times\dots\times\mathbf{n}_{i_{m-1}}|\dots|\mathbf{n}_{i_{m-1}}\times \mathbf{n}_{i_m}\times \mathbf{n}_{i_{m-1}}\dots\times\mathbf{n}_{i_{m-3}}|},\\
&=\frac{-\left|\begin{array}{ccc}
\mathbf{n}_{i_{1}} & \dots & \mathbf{n}_{i_{m-1}}\\
\mathbf{n}_{i_{1}}\cdot(\mathbf{n}_{i_{2}}\times\dots\times\mathbf{n}_{i_{m}})& \dots & \mathbf{n}_{i_{m-1}}\cdot(\mathbf{n}_{i_{2}}\times\dots\times\mathbf{n}_{i_{m}})\\
\vdots & \ddots & \vdots \\
\mathbf{n}_{i_{1}}\cdot(\mathbf{n}_{i_{m-1}}\times\mathbf{n}_{i_{m}}\times\mathbf{n}_{i_{1}}\times\dots\times\mathbf{n}_{i_{m-3}}) & \dots &
\mathbf{n}_{i_{m-1}}\cdot(\mathbf{n}_{i_{m-1}}\times\mathbf{n}_{i_{m}}\times\mathbf{n}_{i_{1}}\times\dots\times\mathbf{n}_{i_{m-3}})
\end{array}\right|}
{\sin\left(\Theta^{[m-1]}_{i_1\dots i_{m-1}}\right)\dots\sin\left(\Theta^{[m-1]}_{i_{m-1}i_mi_1\dots i_{m-3}}\right)},\\
&=\frac{(\mathbf{n}_{i_{1}}\cdot(\mathbf{n}_{i_{2}}\times\dots\times\mathbf{n}_{i_{m}}))^{m-2}\mathbf{n}_{i_{m-1}}}{\sin\left(\Theta^{[m-1]}_{i_1\dots i_{m-1}}\right)\dots\sin\left(\Theta^{[m-1]}_{i_{m-1}i_mi_1\dots i_{m-3}}\right)}.\\
\end{split}
\end{equation}
Hence,
\begin{equation}
\begin{split}
|\mathbf{u}_{i_1\dots i_{m-1}}\times\mathbf{u}_{i_2\dots i_{m}}\times\dots\times\mathbf{u}_{i_{m-1}i_mi_1\dots i_{m-3}}|&=\frac{(\mathbf{n}_{i_{1}}\cdot(\mathbf{n}_{i_{2}}\times\dots\times\mathbf{n}_{i_{m}}))^{m-2}}{\sin(\Theta^{[m-1]}_{i_1\dots i_{m-1}})\dots\sin\left(\Theta^{[m-1]}_{i_{m-1}i_mi_1\dots i_{m-3}}\right)},\\
&=\frac{\sin^{m-2}\left(\Theta^{[1]}_{i_1\dots i_m}\right)}{\sin\left(\Theta^{[m-1]}_{i_1\dots i_{m-1}}\right)\dots\sin\left(\Theta^{[m-1]}_{i_{m-1}i_mi_1\dots i_{m-3}}\right)},
\end{split}
\end{equation}
and so, the sine rule follows.
\end{proof}

\begin{prop}[Polar Cosine Rule]
\begin{equation}
\cos\theta^{[1]}_{i_{m-1}i_{m}}=\mathrm{det}(X),
\end{equation}
where
\begin{equation}\label{eq:X} 
(X)_{jk}=\left\{\begin{array}{cc}
-1; & j=k-1,\\
-\cos\theta^{[m-1]}_{i_1\dots i_{j-1}i_{j+1}\dots i_{k-2}i_k\dots i_m}; & j\neq k-1,\; m \;\mathrm{odd},\\
(-1)^{j+k+1}\cos\theta^{[m-1]}_{i_1\dots i_{j-1}i_{j+1}\dots i_{k-2}i_k\dots i_m}; & j\neq-1,\; m\; \mathrm{even}.
\end{array}\right.
\end{equation}
\end{prop}

\begin{proof}
Consider the product
\begin{equation}
\begin{split}
&(\mathbf{u}_{i_{1}\dots i_{m-1}}\times\mathbf{u}_{i_{2}\dots i_{m}}\times\mathbf{u}_{i_{3}\dots i_{m}i_1}\times\dots\times\mathbf{u}_{i_{m-1}i_mi_1\dots i_{m-3}})\cdot (\mathbf{u}_{i_{2}\dots i_{m}}\times\dots\times\mathbf{u}_{i_{m}i_1\dots i_{m-2}})\\
&=\left|\begin{array}{ccc}
\mathbf{u}_{i_{1}\dots i_{m-1}}\cdot\mathbf{u}_{i_{2}\dots i_{m}} & \dots & \mathbf{u}_{i_{1}\dots i_{m-1}}\cdot\mathbf{u}_{i_mi_{1}\dots i_{m-2}}\\
\vdots&\ddots&\vdots\\
\mathbf{u}_{i_{m-1}i_mi_1\dots i_{m-3}}\cdot\mathbf{u}_{i_{2}\dots i_{m}} & \dots & \mathbf{u}_{i_{m-1}i_mi_1\dots i_{m-3}}\cdot\mathbf{u}_{i_{m}i_1\dots i_{m-2}}
\end{array}\right|,\\
&=\mathrm{det}(X),
\end{split}
\end{equation}
where $X$ is as defined in \eqref{eq:X}. .
Alternatively, recalling (\ref{eq:84}) gives
\begin{equation}
\begin{split}
&(\mathbf{u}_{i_{1}\dots i_{m-1}}\times\mathbf{u}_{i_{2}\dots i_{m}}\times\mathbf{u}_{i_{3}\dots i_{m}i_1}\times\dots\times\mathbf{u}_{i_{m-1}i_mi_1\dots i_{m-3}})\cdot (\mathbf{u}_{i_{2}\dots i_{m}}\times\dots\times\mathbf{u}_{i_{m}i_1\dots i_{m-2}}),\\
&=\left(\frac{(\mathbf{n}_{i_{1}}\cdot(\mathbf{n}_{i_{2}}\times\dots\times\mathbf{n}_{i_{m}}))^{m-2}\mathbf{n}_{i_{m-1}}}{\sin\left(\Theta^{[m-1]}_{i_1\dots i_{m-1}}\right)\dots\sin\left(\Theta^{[m-1]}_{i_{m-1}i_mi_1\dots i_{m-3}}\right)}\right)\cdot\left(\frac{(\mathbf{n}_{i_{2}}\cdot(\mathbf{n}_{i_{3}}\times\dots\times\mathbf{n}_{i_{m}}\times\mathbf{n}_{i_{m1}}))^{m-2}\mathbf{n}_{i_{m}}}{\sin\left(\Theta^{[m-1]}_{i_2\dots i_{m}}\right)\dots\sin\left(\Theta^{[m-1]}_{i_{m}i_1\dots i_{m-2}}\right)}\right),\\
&=\frac{(\mathbf{n}_{i_{1}}\cdot(\mathbf{n}_{i_{2}}\times\dots\times\mathbf{n}_{i_{m}}))^{2m-4}\cos\theta^{[1]}_{i_{m-1}i_m}}{\sin(\Theta^{[m-1]}_{i_1\dots i_{m-1}})\sin^2\left(\Theta^{[m-1]}_{i_2\dots i_{m}}\right)\dots\sin^2\left(\Theta^{[m-1]}_{i_{m-1}i_mi_1\dots i_{m-3}}\right)\sin\left(\Theta^{[m-1]}_{i_mi_1\dots i_{m-2}}\right)},\\
&=\frac{\sin^{2m-4}(\Theta^{[1]}_{i_1\dotsi_m})\cos\theta^{[1]}_{i_{m-1}i_m}}{\sin(\Theta^{[m-1]}_{i_1\dots i_{m-1}})\sin^2\left(\Theta^{[m-1]}_{i_2\dots i_{m}}\right)\dots\sin^2\left(\Theta^{[m-1]}_{i_{m-1}i_mi_1\dots i_{m-3}}\right)\sin\left(\Theta^{[m-1]}_{i_mi_1\dots i_{m-2}}\right)}.\\
\end{split}
\end{equation}
By applying the sine rule, the Polar Cosine rule then follows.
\end{proof}

Note that the facets of the $(m-1)$-dimensional simplex are $(m-2)$-dimensional simplices with the various cosine and sine rules still holding true in these facets. The same applies to these $(m-2)$-dimensional facets, and so on, forming a  complete hierarchy of cosine and sine rules between the angles for the $j$-dimensional facets, and those for the $(j-1)$-dimensional facets. So, in terms of the cosine rule we have
\begin{equation}
\cos\theta^{[j]}_{i_2\dots i_{j-1}}=\frac{-\left|\begin{array}{cc}
\cos\theta^{[j-1](i_1\dots i_{j-1})}_{i_2\dots i_{j-1}} & \cos\theta^{[j-1](i_1\dots i_{j-2}i_j)}_{i_2\dots i_{j-1}}\\
1 & \cos\theta^{[j-1](i_2\dots i_{j})}_{i_2\dots i_{j-1}}\\
\end{array}\right|}
{\sin\theta^{[j-1](i_1\dots i_{j-1})}_{i_2\dots i_{j-1}}\sin\theta^{[j-1](i_2\dots i_{j})}_{i_2\dots i_{j-1}}}.
\end{equation}
This can be extended so that the angles of any facet can be expressed in terms of those for any other dimensional facets to give a cosine rule of the form
\begin{equation}
\cos\theta^{[j]}_{i_2\dots i_{j-1}}=\frac{-\left|\begin{array}{cccc}
\cos\theta^{[k](i_1\dots i_ki_{k+1})}_{i_2\dots i_k} & \cos\theta^{[k](i_1\dots i_ki_{k+2})}_{i_2\dots i_k} & \dots &\cos\theta^{[k](i_1\dots i_ki_{j})}_{i_2\dots i_k}\\
\cos\theta^{[k](i_2\dots i_ki_{k+1})}_{i_2\dots i_{k+1}i_{k+1}} & \cos\theta^{[k](i_2\dots i_{k+2})}_{i_2\dots i_k} & \dots &\cos\theta^{[k](i_2\dots i_{k+1}i_{j})}_{i_2\dots i_k}\\
\vdots & \vdots &\ddots & \vdots\\
\cos\theta^{[k](i_{j-1}i_2\dots i_{k+1})}_{i_2\dots i_k} & \cos\theta^{[k](i_{j-1}i_2\dots i_ki_{k+2})}_{i_2\dots i_k} & \dots &\cos\theta^{[k](i_{j-1}i_2\dots i_ki_{j})}_{i_2\dots i_k}\\
\end{array}\right|}
{\sin\left(\Theta^{[k](i_1\dots i_k)}_{i_2\dots i_k}\right)\sin\left(\Theta^{[k](i_2\dots i_{k+1})}_{i_2 \dots i_k}\right)}.
\end{equation}
From these cosine rules, there arise corresponding sine rules. Specifically, note that
\begin{equation}
\sin\theta^{[j]}_{i_2\dots i_{j}}=k_{j-1}\sin\theta^{[j-1]}_{i_2\dots i_{j-1}},
\end{equation}
where
\begin{equation}
k_{j-1}=\frac{\sin\left(\theta^{[j-1](i_1\dots i_j)}_{i_2\dots i_{j-1}},\theta^{[j-1](i_2\dots i_{j+1})}_{i_2\dots i_{j-1}},\theta^{[j-1](i_1\dots i_{j-1}i_{j+1})}_{i_2\dots i_{j-1}}
\right)}{\sin\left(\theta^{[j-1](i_1\dots i_j)}_{i_2\dots i_{j-1}}\right)\sin\left(\theta^{[j-1](i_2\dots i_{j+1})}_{i_2\dots i_{j-1}}\right)\sin\left(\theta^{[j-1](i_1\dots i_{j-1}i_{j+1})}_{i_2\dots i_{j-1}}\right)}\ ,
\end{equation}
with $k_{j-1}$ being constant. Hence, from this, it follows that a full hierarchy of intertwined sine rules exists between the angles of any two different dimensional facets.

\section{Link between Hyperspherical Trigonometry and Elliptic Functions}
\subsection{Link between Spherical Trigonometry and Jacobi Elliptic Functions}
There exists a well-known link between the formulae of spherical trigonometry and the Jacobi elliptic functions through their addition formulae, \cite{Lagrange,Legendre}. Here, we follow a derivation given by Irwin\cite{irwin}, but with a slight modification, taking as a starting point the sine rule for spherical trigonometry,
\begin{equation}
\frac{\sin\alpha_i}{\sin\theta_{jk}}=\frac{\sin\alpha_j}{\sin\theta_{ik}}=\frac{\sin\alpha_k}{\sin\theta_{ij}}=k,
\end{equation}
where
\begin{equation}\label{eq:95}
k=\frac{\sin(\theta_{ij},\theta_{ik},\theta_{jk})}{\sin\theta_{ij}\sin\theta_{ik}\sin\theta_{jk}},
\end{equation}
and considering the expression
\begin{equation}
W=k^2\sin^2\theta_{ij}\sin^2\theta_{ik}\sin^2\theta_{jk}-\sin^2(\theta_{ij},\theta_{ik},\theta_{jk}).
\end{equation}
It follows that the derivative of $W$ with respect to one of the angles, $\theta_{ij}$, is
\begin{equation}
\begin{split}
\frac{\partial W}{\partial \theta_{ij}}&=-2\sin\theta_{ij}(\cos\theta_{ij}-\cos\theta_{ik}\cos\theta_{jk}-k^2\cos\theta_{ij}\sin^2\theta_{ik}\sin^2\theta_{jk}),\\
&=2\sin\theta_{ij}\sin\theta_{ik}\sin\theta_{jk}\cos\alpha_i\cos\alpha_j.
\end{split}
\end{equation}
If the radial angles $\theta_{ij}$, $\theta_{ik}$ and $\theta_{jk}$ all vary in such a way as to keep $k$ constant, then since this is equivalent to $W=0$, this implies
\begin{equation}\label{eq:diffrel} 
\cos\alpha_i\cos\alpha_j\mathrm{d}\theta_{ij}+\cos\alpha_i\cos\alpha_k\mathrm{d}\theta_{ik}+\cos\alpha_j\cos\alpha_k\mathrm{d}\theta_{jk}=0,
\end{equation}
which is of course equivalent to
\begin{equation}
\frac{\mathrm{d\theta_{ij}}}{\cos\alpha_k}+\frac{\mathrm{d\theta_{jk}}}{\cos\alpha_i}+\frac{\mathrm{d\theta_{ik}}}{\cos\alpha_j}=0.
\end{equation}
However, the sine rule gives
\begin{equation}
\cos\alpha_i=\pm\sqrt{1-k^2\sin^2\theta_{jk}},\, \mathrm{etc.},
\end{equation}
and so, as the signs of $\theta_{ij}$, $\theta_{ik}$ and $\theta_{jk}$ can be chosen arbitrarily, without loss of generality they can all be chosen to be positive, giving
\begin{equation}
\label{eq:int}
\frac{\mathrm{d\theta_{ij}}}{\sqrt{1-k^2\sin^2\theta_{ij}}}+\frac{\mathrm{d\theta_{ik}}}{\sqrt{1-k^2\sin^2\theta_{ik}}}+\frac{\mathrm{d\theta_{jk}}}{\sqrt{1-k^2\sin^2\theta_{jk}}}=0. 
\end{equation}
Writing $\theta_{ij}=\mathrm{am}(a_k)$, with
\begin{equation}
a_k=\int_{0}^{\mathrm{am}(a_k)}\frac{\mathrm{d}t}{\sqrt{1-k^2\sin^2t}},
\end{equation}
and similarly,  $\theta_{ik}=\mathrm{am}(a_j)$ and  $\theta_{jk}=\mathrm{am}(a_i)$, the integral of the differential relation, \eqref{eq:int}, implies that
\begin{equation}
a_i+a_j+a_k=\delta,\, \mathrm{constant},
\end{equation}
together with
\begin{equation}
\begin{split}
&\sin\theta_{ij}=\sin(\mathrm{am}(a_k))=\mathrm{sn}(a_k),\\
&\sin\theta_{ik}=\sin(\mathrm{am}(a_j))=\mathrm{sn}(a_j),\\
&\sin\theta_{jk}=\sin(\mathrm{am}(a_i))=\mathrm{sn}(a_i).
\end{split}
\end{equation}
Therefore, having introduced uniformising variables $a_i$, $i=1,2,3,$ associated with the three spherical angles $\theta_{jk}$, the various spherical trigonometric functions can be identified with the Jacobi elliptic functions by using the identifications
\begin{equation}\label{eq:106}
\sin(\theta_{jk})=\text{sn}(a_i;k) \iff \sin\alpha_i\equiv k\,\text{sn}(a_i;k),
\end{equation}
in which $k$ is the modulus of the elliptic function given by (\ref{eq:95}). These identifications, through the usual relations between the three Jacobi elliptic functions $\mathrm{sn}$,  $\mathrm{cn}$ and  $\mathrm{dn}$,
\begin{equation}\begin{array}{cc}
\mathrm{cn}^2(u;k)+\mathrm{sn}^2(u;k)=1, \hspace{10mm}\mathrm{dn}^2(u;k)+k^2\mathrm{sn}^2(u;k)=1,\\
\end{array}
\end{equation}
lead to
\begin{equation}\label{eq:108}
\cos\theta_{jk}=\text{cn}(a_i;k),\,\mathrm{and}\,\cos\alpha_i=\text{dn}(a_i;k),
\end{equation}
for $i,j,k$ cyclic.
Addition formulae follow readily from the various spherical trigonometric relations. In particular, the cosine and polar cosine rules yield the relations
\begin{equation}
\begin{array}{c}
\mathrm{cn}(a_i)=\mathrm{cn}(a_j)\mathrm{cn}(a_k)+\mathrm{sn}(a_j)\mathrm{sn}(a_k)\mathrm{dn}(a_i),\\
\mathrm{dn}(a_i)=-\mathrm{dn}(a_j)\mathrm{dn}(a_k)+k^2\mathrm{sn}(a_j)\mathrm{sn}(a_k)\mathrm{cn}(a_i),
\end{array}
\end{equation}
respectively. Solving these relations as functions of $a_i$ gives 
\begin{equation}
\mathrm{cn}(a_i)=\frac{\mathrm{dn}(a_j)\mathrm{dn}(a_k)\mathrm{sn}(a_j)\mathrm{sn}(a_k)-\mathrm{cn}(a_j)\mathrm{cn}(a_k)}{1-k^2\mathrm{sn}^2(a_j)\mathrm{sn}^2(a_k)},
\end{equation}
\begin{equation}
\mathrm{dn}(a_i)=\frac{\mathrm{dn}(a_j)\mathrm{dn}(a_k)-k^2\mathrm{cn}(a_j)\mathrm{cn}(a_k)\mathrm{sn}(a_j)\mathrm{sn}(a_k)}{1-k^2\mathrm{sn}^2(a_j)\mathrm{sn}^2(a_k)},
\end{equation}
the addition formulae for the Jacobi elliptic functions with
\begin{equation}
\begin{array}{c}
\mathrm{cn}(a_i)=-\mathrm{cn}(a_j+a_k),\\
\mathrm{dn}(a_i)=\mathrm{dn}(a_j+a_k).\\
\end{array}
\end{equation}
The Jacobi elliptic functions are periodic in $K(k)$ and $K'(k)$ as
\begin{equation}
\begin{array}{c}
\mathrm{cn}(a_j+2mK+2niK';k)=(-1)^m\mathrm{cn}(a_j;k),\\
\mathrm{dn}(a_j+2mK+2niK';k)=(-1)^n\mathrm{dn}(a_j;k),\\
\end{array}
\end{equation}
for all $a_j$, and so we must restrict $\delta$ such that $a_i+a_j+a_k=2K+2iK'$. The addition formula for $\mathrm{sn}$ follows as a consequence. From these addition formulae a number of intertwined addition relations follow
\begin{equation}
\text{cn}(a_j)\text{sn}(a_i+a_j)=\text{sn}(a_i)\text{dn}(a_j)+\text{dn}(a_i)\text{sn}(a_j)\text{cn}(a_i+a_j),
\end{equation}
\begin{equation}
\text{dn}(a_i)\text{sn}(a_i+a_j)=\text{cn}(a_i)\text{sn}(a_j)+\text{sn}(a_i)\text{cn}(a_j)\text{dn}(a_i+a_j),
\end{equation}
\begin{equation}
\text{sn}(a_i)\text{cn}(a_i+a_j)+\text{sn}(a_j)\text{dn}(a_i+a_j)=\text{cn}(a_i)\text{dn}(a_j)\text{sn}(a_i+a_j).
\end{equation}
Denoting
\begin{equation}
w_1(a_j)=\frac{\rho}{\text{sn}(a_j)},\,w_2(a_j)=\frac{\rho\,\text{cn}(a_j)}{\text{sn}(a_j)},\,w_3(a_j)=\frac{\rho\,\text{dn}(a_j)}{\text{sn}(a_j)},
\end{equation}
these addition formulae may be rewritten as
\begin{equation}
w_i(a_j)w_j(a_i)+w_j(a_k)w_k(a_j)+w_k(a_i)w_i(a_k)=0,\, \,i,j,k=1,2,3,
\end{equation}
with $a_1+a_2+a_3=2K+2iK'$. Note that this relation is the functional Yang-Baxter relation\cite{Baxter}.

Now, by considering the reciprocal of $k$,
\begin{equation}
\frac{\sin\theta_{jk}}{\sin\alpha_i}=\frac{\sin\theta_{ik}}{\sin\alpha_j}=\frac{\sin\theta_{ij}}{\sin\alpha_k}=\frac{1}{k},
\end{equation}
following Irwin's method, we can derive a similar, yet slightly simpler, relationship. First, recall
\begin{equation}
k=\frac{\sin\alpha_k}{\sin\theta_{ij}}=\frac{\sin\alpha_i\sin\alpha_j\sin\alpha_k}{\sqrt{1-\cos^2\alpha_i-\cos^2\alpha_j-\cos^2\alpha_k-2\cos\alpha_i\cos\alpha_j\cos\alpha_k}},
\end{equation}
and let
\begin{equation}
\overline{W}=\frac{1}{k^2}\sin^2\alpha_i\sin^2\alpha_j\sin^2\alpha_k-(1-\cos^2\alpha_i-\cos^2\alpha_j-\cos^2\alpha_k-2\cos\alpha_i\cos\alpha_j\cos\alpha_k).
\end{equation}
Now, consider the derivative $\partial\overline{W}/\partial\alpha_i$,
\begin{equation}
\partial\overline{W}/\partial\alpha_i=-2\cos\theta_{ij}\cos\theta_{ik}\sin\alpha_i\sin\alpha_j\sin\alpha_k.
\end{equation}
Varying $\alpha_i$, $\alpha_j$ and $\alpha_k$ as to keep $k$  constant, it follows that
\begin{equation}\label{eq:int'}
\frac{\mathrm{d}\alpha_i}{\sqrt{1-\frac{1}{k^2}\sin^2\alpha_i}}+\frac{\mathrm{d}\alpha_j}{\sqrt{1-\frac{1}{k^2}\sin^2\alpha_i}}+\frac{\mathrm{d}\alpha_k}{\sqrt{1-\frac{1}{k^2}\sin^2\alpha_i}}=0.
\end{equation}
Writing $\alpha_i=\mathrm{am}(b_i)$, with
\begin{equation}
b_i=\int_{0}^{\mathrm{am}(b_i)}\frac{\mathrm{d}t}{\sqrt{1-\frac{1}{k^2}\sin^2t}},
\end{equation}
and similarly, $\alpha_j=\mathrm{am}(b_j)$ and $\alpha_k=\mathrm{am}(b_k)$, the integral of (\ref{eq:int'}) implies
\begin{equation}
b_i+b_j+b_k=\gamma,\,\,\mathrm{constant},
\end{equation}
together with
\begin{equation}
\begin{split}
&\sin\alpha_i=\sin(\mathrm{am}(b_i))=\mathrm{sn}(b_i),\\
&\sin\alpha_j=\sin(\mathrm{am}(b_j))=\mathrm{sn}(b_j),\\
&\sin\alpha_k=\sin(\mathrm{am}(b_k))=\mathrm{sn}(b_k).\\
\end{split}
\end{equation}
Therefore, having introduced spherical angles $b_i$, $i=1,2,3$, this time associated with the three spherical angles $\alpha_i$, the various spherical trigonometric functions can be identified with the Jacobi elliptic functions via the identifications
\begin{equation}
\sin\alpha_i=\mathrm{sn}\left(b_i;\frac{1}{k}\right)\iff\sin\theta_{jk}=\frac{1}{k}\mathrm{sn}\left(b_i;\frac{1}{k}\right),
\end{equation}
in which $1/k$ is the modulus of the elliptic function. These identifications lead to
\begin{equation}
\cos\alpha_i=\mathrm{cn}\left(b_i;\frac{1}{k}\right),\,\,\mathrm{and}\,\,\cos\theta_{jk}=\mathrm{dn}\left(b_i;\frac{1}{k}\right).
\end{equation}
In this case, the Jacobi elliptic function addition formulae are satisfied providing we have
\begin{equation}
\begin{split}
&\mathrm{cn}(b_i)=\mathrm{cn}(b_j+b_k),\\
&\mathrm{dn}(b_i)=-\mathrm{dn}(b_j+b_k).
\end{split}
\end{equation}
The periodicity ensures $\gamma=2K$,
\begin{equation}
b_i+b_j+b_k=2K.
\end{equation}
Note that making the identification in this way results in a slightly simpler restriction, dependent only on $K$, and not also $K'$.

\subsection{Link between Hyperspherical Trigonometry and Elliptic Functions}

We now look to provide a similar link for the 4 dimensional hyperspherical case and elliptic functions following a similar procedure as before. We take as a starting point the hyperspherical sine rule relation,
\begin{equation}\label{eq:hypsin} 
\frac{\sin(\phi_{ik},\phi_{jk},\phi_{kl})}{\sin(\theta_{ij},\theta_{il},\theta_{jl})}=\frac{\sin^2(\theta_{ij},\theta_{ik},\theta_{il},\theta_{jk},\theta_{jl},\theta_{kl})}{\sin(\theta_{ij},\theta_{ik},\theta_{jk})\sin(\theta_{ij},\theta_{il},\theta_{jl})\sin(\theta_{ik},\theta_{il},\theta_{kl})\sin(\theta_{jk},\theta_{jl},\theta_{kl})}=k,
\end{equation}
from which we introduce
\begin{align}
W &=\sin^4(\theta_{ij},\theta_{ik},\theta_{il},\theta_{jk},\theta_{jl},\theta_{kl}) \nonumber \\ 
& \qquad -k^2\sin^2(\theta_{ij},\theta_{ik},\theta_{jk})\sin^2(\theta_{ij},\theta_{il},\theta_{jl})\sin^2(\theta_{ik},\theta_{il},\theta_{kl})\sin^2(\theta_{jk},\theta_{jl},\theta_{kl}).
\end{align}
Taking $W$'s derivative with respect to $\theta_{ij}$ gives
\begin{equation}\label{eq:diff}
\begin{split}
\frac{\partial W}{\partial\theta_{ij}}=&2\sin^2(\theta_{ij},\theta_{ik},\theta_{il},\theta_{jk},\theta_{jl},\theta_{kl})\left(\frac{\partial\sin^2(\theta_{ij},\theta_{ik},\theta_{il},\theta_{jk},\theta_{jl},\theta_{kl})}{\partial\theta_{ij}}\right)\\
&-k^2\left(\frac{\partial\sin^2(\theta_{ij},\theta_{ik},\theta_{jk})}{\partial\theta_{ij}}\right)\sin^2(\theta_{ij},\theta_{il},\theta_{jl})\sin^2(\theta_{ik},\theta_{il},\theta_{kl})\sin^2(\theta_{jk},\theta_{jl},\theta_{kl})\\
&-k^2\sin^2(\theta_{ij},\theta_{ik},\theta_{jk})\left(\frac{\partial\sin^2(\theta_{ij},\theta_{il},\theta_{jl})}{\partial\theta_{ij}}\right)\sin^2(\theta_{ik},\theta_{il},\theta_{kl})\sin^2(\theta_{jk},\theta_{jl},\theta_{kl}).
\end{split}
\end{equation}
Considering the derivatives of the squares of the triple sine functions, we have
\begin{equation}
\begin{split}
\frac{\partial\sin^2(\theta_{ij},\theta_{ik},\theta_{jk})}{\partial\theta_{ij}}&=2\cos\theta_{ij}\sin\theta_{ij}-2\sin\theta_{ij}\cos\theta_{ik}\cos\theta_{jk}\\
&=2\cos\alpha_k^{(ijk)}\sin\theta_{ij}\sin\theta_{ik}\sin\theta_{jk},
\end{split}
\end{equation}
and similarly,
\begin{equation}
\frac{\partial\sin^2(\theta_{ij},\theta_{il},\theta_{jl})}{\partial\theta_{ij}}=2\cos\alpha_l^{(ijl)}\sin\theta_{ij}\sin\theta_{il}\sin\theta_{jl}.
\end{equation}
We also have, for the six term sine function,
\begin{equation}
\frac{\partial\sin^2(\theta_{ij},\theta_{ik},\theta_{il},\theta_{jk},\theta_{jl},\theta_{kl})}{\partial\theta_{ij}}
=-2\sin\theta_{ij}\cos\phi_{kl}\sin(\theta_{ik},\theta_{il},\theta_{kl})\sin(\theta_{jk},\theta_{jl},\theta_{kl}).
\end{equation}
Substituting these back into (\ref{eq:diff}) gives
\begin{equation}\label{eq:diff'}
\begin{split}
\frac{\partial W}{\theta_{ij}}=&-2\sin\theta_{ij}\sin^2(\theta_{ij},\theta_{ik},\theta_{il},\theta_{jk},\theta_{jl},\theta_{kl})\\
&\times\Bigg(2\cos\phi_{kl}\sin(\theta_{ik},\theta_{il},\theta_{kl})\sin(\theta_{jk},\theta_{jl},\theta_{kl})\\
&+\sin^2(\theta_{ij},\theta_{ik},\theta_{il},\theta_{jk},\theta_{jl},\theta_{kl})\left(
\frac{\cos\alpha_k^{(ijk)}\sin\theta_{ik}\sin\theta_{jk}}{\sin^2(\theta_{ij},\theta_{ik},\theta_{jk})}+\frac{\cos\alpha_l^{(ijl)}\sin\theta_{il}\sin\theta_{jl}}{\sin^2(\theta_{ij},\theta_{il},\theta_{jl})}
\right)\Bigg).
\end{split}
\end{equation}
In order to simplify this derivative, we now consider the Desnanot-Jacobi determinantal identity,
\begin{equation}
\begin{split}
&\left|\begin{array}{cc}
\cos\theta_{jk} & \cos\theta_{jl}\\
\cos\theta_{kl} & 1\\
\end{array}\right|
\left|\begin{array}{cccc}
1 &  \cos\theta_{ij} & \cos\theta_{ik} & \cos\theta_{il}\\
\cos\theta_{ij} & 1 & \cos\theta_{jk} & \cos\theta_{jl}\\
\cos\theta_{ik} & \cos\theta_{jk} & 1 & \cos\theta_{kl}\\
\cos\theta_{il} & \cos\theta_{jl} & \cos\theta_{kl} & 1\\
\end{array}\right|\\
&\qquad =\left|\begin{array}{ccc}
1 & \cos\theta_{jk} & \cos\theta_{jl}\\
\cos\theta_{jk} & 1 & \cos\theta_{kl}\\
\cos\theta_{jl} & \cos\theta_{kl} & 1\\
\end{array}\right|
\left|\begin{array}{ccc}
1 & \cos\theta_{ik} & \cos\theta_{il}\\
\cos\theta_{ij} & \cos\theta_{jk} & \cos\theta_{jl}\\
\cos\theta_{il} & \cos\theta_{kl} & 1\\
\end{array}\right| \nonumber \\ 
& \qquad\qquad -
\left|\begin{array}{ccc}
\cos\theta_{ij} & \cos\theta_{ik} & \cos\theta_{il}\\
1 & \cos\theta_{jk} & \cos\theta_{jl}\\
\cos\theta_{jl} & \cos\theta_{kl} & 1\\
\end{array}\right|
\left|\begin{array}{ccc}
\cos\theta_{ij} & \cos\theta_{jk} & \cos\theta_{jl}\\
\cos\theta_{ik} & 1 & \cos\theta_{kl}\\
\cos\theta_{il} & \cos\theta_{kl} & 1\\
\end{array}\right|,
\end{split}
\end{equation}
which, reduces neatly to
\begin{align}
& \frac{\sin^2(\theta_{ij},\theta_{ik},\theta_{il},\theta_{jk},\theta_{jl},\theta_{kl})\cos\alpha_l^{(jkl)}\sin\theta_{jl}\sin\theta_{kl}}{\sin^2(\theta_{jk},\theta_{jl},\theta_{kl})} 
\nonumber \\ 
& \qquad =
\sin(\theta_{ij},\theta_{il},\theta_{jl})\sin(\theta_{ik},\theta_{il},\theta_{kl})\left(\cos\phi_{jl}\cos\phi_{kl}-
\cos\phi_{il}\right).
\end{align}
Substituting this expression into (\ref{eq:diff'}), the derivative reduces to 
\begin{align}
\frac{\partial W}{\partial\theta_{ij}}& =-2\sin\theta_{ij}\sin^2(\theta_{ij},\theta_{ik},\theta_{il},\theta_{jk},\theta_{jl},\theta_{kl}) \nonumber \\ 
& \qquad \times \sin(\theta_{ik},\theta_{il},\theta_{kl})\sin(\theta_{jk},\theta_{jl},\theta_{kl})\left(\cos\phi_{il}\cos\phi_{jl}+\cos\phi_{ik}\cos\phi_{jk}\right),
\end{align}
which, using (\ref{eq:sinn}), may be rewritten as
\begin{equation}
\begin{split}
\frac{\partial W}{\partial\theta_{ij}}& =-2\sin(\theta_{ij},\theta_{ik},\theta_{il},\theta_{jk},\theta_{jl},\theta_{kl})
\sin(\theta_{ik},\theta_{il},\theta_{kl})\sin(\theta_{jk},\theta_{jl},\theta_{kl})  \\ 
& \quad \times \sin(\theta_{ij},\theta_{il},\theta_{jl})\sin(\theta_{ij},\theta_{ik},\theta_{jk})
\cdot\sin\phi_{ij}
\left(\cos\phi_{il}\cos\phi_{jl}+\cos\phi_{ik}\cos\phi_{jk}\right).
\end{split}
\end{equation}

If we vary the six `$\theta$'s such that k remains constant, then since this is equivalent to $W=0$, setting $\mathrm{d}W=0$ implies
\begin{equation}\label{eq:perm} 
\sum_{perm}\sin\phi_{ij}\left(\cos\phi_{il}\cos\phi_{jl}+\cos\phi_{ik}\cos\phi_{jk}\right)\mathrm{d}\theta_{ij}=0, 
\end{equation}
where $perm$ denotes all six index pairs $ij,ik,il,jk,jl,kl$. In spite of the similarity between \eqref{eq:perm} and the 
analogous formula \eqref{eq:diffrel} in the spherical case, the former cannot be simplified along the same line 
as in the spherical case. A way out, albeit not quite satisfactory, is to  
focus on a special symmetric case, in which the opposite dihedral angles are equivalent, i.e. 
\[ \phi_{ij}=\phi_{kl}, \quad \phi_{ik}=\phi_{jl}, \quad \phi_{il}=\phi_{jk}\ ,  \]
which we will refer to as the \textit{symmetric hyperspherical tetrahedron}.  

Restricting ourselves to this symmetric case, the differential relation \eqref{eq:perm} reduces to 
\begin{equation}\label{eq:symperm}
\sum_{perm} \sin\phi_{ij}\,\cos\phi_{ik}\,\cos\phi_{jk}\,\mathrm{d}\theta_{ij}=0, 
\end{equation}   
or equivalently 
\begin{equation}\label{eq:symperm2}
\sum_{perm} \frac{\sin\phi_{ij}}{\cos\phi_{ij}}\,\mathrm{d}\theta_{ij}=0.  
\end{equation} 
The restrictions on the tetrahedron being symmetric also simplify the hyperspherical sine rule 
\eqref{eq:sin}, reducing it to 
\begin{equation}\label{eq:redsin} 
\sin^2\phi_{ij}=k_H\sin^2\theta_{ij}\ , 
\end{equation} 
and substituting this into \eqref{eq:symperm2} we obtain 
\begin{equation}\label{eq:symperm3}
\sum_{perm} \frac{\sqrt{k_H}\,\sin\theta_{ij}}{\sqrt{1-k_H \sin^2\theta_{ij}}}\,\mathrm{d}\theta_{ij}=0,  
\end{equation} 
which yields a sum of elliptic integrals of the form 
\[ \int \frac{\sqrt{k_H}\,u}{\sqrt{(1-k_Hu^2)(1-u^2)}}\,\mathrm{d}u\ . \] 
Thus, for the case of a symmetric hyperspherical tetrahedron, the hyperspherical trigonometric functions 
which in a sense uniformise the hyperspherical sine rule \eqref{eq:hypsin} are again associated with elliptic 
functions, suggesting that for the general case this remains true. However, in that general case the 
connection becomes a bit more entangled and require the introduction of a generalised class of elliptic functions, 
introduced by Pawellek, \cite{pawellek}.

\section{Generalised Jacobi Elliptic Functions and their link to Hyperspherical Trigonometry}

In this section we review a class of `generalised' Jacobi elliptic functions as introduced by Pawellek\cite{pawellek}, and provide a link between these functions and the formulae of hyperspherical trigonometry. 
%\footnote{In actual fact, these functions which 
%in appearance seem to be related to a genus $g=2$ curve,  
%are still deep down expressible in terms of elliptic functions, and thus represent some cover of the 
%elliptic curve. Nevertheless the connection with hyperspherical trigonometry, and the special role of the moduli, 
%show that there is some advantage in studying these functions on their own merit.}

\subsection{Generalised Jacobi Elliptic Functions}
In \cite{pawellek}, Pawellek introduced the generalised Jacobi elliptic functions $s(u,k_1,k_2)$, $c(u,k_1,k_2)$, $d_1(u,k_1,k_2)$ and $d_2(u,k_1,k_2)$. These functions are algebraic curves with two distinct moduli, and are based upon Jacobi's elliptic functions, \cite{Jacobi}. After assuming without loss of generality that $1>k_1>k_2>0$ as moduli parameters, they are defined as the inversion of the hyperelliptic integrals 
\bse \begin{equation}
u(x,k_1,k_2)=\int_0^{x=s(u)}\frac{\mathrm{d}t}{\sqrt{(1-t^2)(1-k_1^2t^2)(1-k_2^2t^2)}},
\end{equation}
\begin{equation}
u(x,k_1,k_2)=\int_{x=c(u)}^1\frac{\mathrm{d}t}{\sqrt{(1-t^2)(k_1'^2+k_1^2t^2)(k_2'^2+k_2^2t^2)}},
\end{equation}\begin{equation}
u(x,k_1,k_2)=k_1\int_{x=d_1(u)}^1\frac{\mathrm{d}t}{\sqrt{(1-t^2)(t^2-k_1'^2)(k_1^2-k_2^2+k_2^2t^2)}},
\end{equation}\begin{equation}
u(x,k_1,k_2)=k_2\int_{x=d_2(u)}^1\frac{\mathrm{d}t}{\sqrt{(1-t^2)(t^2-k_2'^2)(k_2^2-k_1^2+k_1^2t^2)}},
\end{equation}\ese 
respectively, with $k_i'=\sqrt{1-k_i^2}$.
These generalised Jacobi functions are associated with an algebraic curve of the form
\begin{equation}
\mathcal{C}:y^2=(1-x^2)(1-k_1^2x^2)(1-k_2^2x^2),
\end{equation}
which can be modeled as a Riemann surface of genus 2, but is, in fact, a double cover of an elliptic curve $\mathcal{E}$,
\begin{equation}
\mathcal{E}:w^2=z(1-z)(1-k_1^2z)(1-k_2^2z),
\end{equation}
via the cover map $\mathcal{C} \overset{\pi}{\rightarrow}\mathcal{E}$,
\begin{equation}
(w,z)=\pi(y,x)=(xy,x^2).
\end{equation}
The functions satisfy a number of identities,
\begin{equation}\label{81}
\begin{split}
&c^2(u)=1-s^2(u),\hspace{10 mm} d_1^2(u)=1-k_1^2s^2(u),\hspace{10 mm} d_2^2(u)=1-k_2^2s^2(u),\\
&d_i^2(u)-k_i^2c^2(u)=1-k_i^2,\hspace{5 mm}i=1,2;\hspace{10 mm}k_1^2d_2^2(u)-k_2^2d_1^2(u)=k_1^2-k_2^2,\\
\end{split}
\end{equation}
and are related to the Jacobi elliptic functions by
\bse \begin{align}
s(u,k_1,k_2)&=\frac{\mathrm{sn}(k_2'u,\kappa)}{\sqrt{k_2'^2+k_2^2\mathrm{sn}^2(k_2'u,\kappa)}}, & 
c(u,k_1,k_2)&=\frac{k_2'\mathrm{cn}(k_2'u,\kappa)}{\sqrt{1-k_2^2\mathrm{cn}^2(k_2'u,\kappa)}},\\
d_1(u,k_1,k_2)&=\frac{\sqrt{k_1^2-k_2^2}\mathrm{dn}(k_2'u,\kappa)}{\sqrt{k_1^2-k_2^2\mathrm{dn}^2(k_2'u,\kappa)}}, & 
d_2(u,k_1,k_2)&=\frac{\sqrt{k_1^2-k_2^2}}{\sqrt{k_1^2-k_2^2\mathrm{dn}^2(k_2'u,\kappa)}},
\end{align}\ese 
with
\begin{equation}
\kappa=\frac{k_1^2-k_2^2}{1-k_2^2}.
\end{equation}
The first derivatives of these functions are given by
\begin{subequations}\label{eq:Pawders} 
\begin{align}
s'(u)&=c(u)d_1(u)d_2(u), & c'(u)&=-s(u)d_1(u)d_2(u),\\
d_1'(u)&=-k_1^2s(u)c(u)d_2(u), & d_2'(u)&=-k_2^2s(u)c(u)d_1(u).
\end{align}
\end{subequations}
These generalised Jacobi functions with moduli $k_1$ and $k_2$ satisfy the following addition formulae:
\bse \begin{equation}\label{85}
s(u\pm v)
=\frac{s(u)d_2(u)c(v)d_1(v)\pm s(v)d_2(v)c(u)d_1(u)}
{\sqrt{\left[d_2^2(u)d_2^2(v)-\kappa^2k_2^{'4}s^2(u)s^2(v)\right]^2+k^2_2\left[s(u)d_2(u)c(v)d_1(v)\pm s(v)d_2(v)c(u)d_1(u)\right]}},
\end{equation}
\begin{equation}\label{86}
c(u\pm v)=\frac{c(u)d_2(u)c(v)d_2(v)\mp k_2^{'2}s(u)d_1(u)s(v)d_1(v)}
{\sqrt{\left[d_2^2(u)d_2^2(v)-\kappa^2k_2^{'4}s^2(u)s^2(v)\right]^2+k^2_2\left[s(u)d_2(u)c(v)d_1(v)\pm s(v)d_2(v)c(u)d_1(u)\right]}},
\end{equation}
\begin{equation}\label{87}
d_1(u\pm v)=\frac{d_1(u)d_2(u)d_1(v)d_2(v)\mp\kappa^2 k_2^{'2}s(u)c(u)s(v)c(v)}
{\sqrt{\left[d_2^2(u)d_2^2(v)-\kappa^2k_2^{'4}s^2(u)s^2(v)\right]^2+k^2_2\left[s(u)d_2(u)c(v)d_1(v)\pm s(v)d_2(v)c(u)d_1(u)\right]}},
\end{equation}
\begin{equation}\label{88}
d_2(u\pm v)=\frac{d_2^2(u)d_2^2(v)-\kappa^2k_2^{'4}s^2(u)s^2(v)}
{\sqrt{\left[d_2^2(u)d_2^2(v)-\kappa^2k_2^{'4}s^2(u)s^2(v)\right]^2+k^2_2\left[s(u)d_2(u)c(v)d_1(v)\pm s(v)d_2(v)c(u)d_1(u)\right]}}.
\end{equation}\ese 
These formulae follow from the addition formulae for the standard Jacobi elliptic functions.

\subsection{Link between Hyperspherical Trigonometry and the Generalised Jacobi Elliptic Functions}

We are also able to show a connection between the generalised Jacobi elliptic functions and hyperspherical trigonometry. As the generalised Jacobi elliptic functions are dependent on two moduli, then so must the hyperspherical trigonometry be. It is the interplay of these moduli that govern the connection. Recall, that for a hyperspherical tetrahedron, we have
\begin{equation}
\sin\alpha_i^{(ijk)}=k_1\sin\theta_{jk},
\end{equation}
and
\begin{equation}
\sin\phi_{il}=k_2\sin\alpha_i^{(ijk)},
\end{equation}
for all $i,j,k,l=1,2,3,4$, where
\begin{equation}
k_1=\frac{\sin(\theta_{ij},\theta_{ik},\theta_{jk})}{\sin\theta_{ij}\sin\theta_{ik}\sin\theta_{jk}},
\end{equation}
and
\begin{equation}
k_2=\frac{\sin(\alpha_i^{(ijk)},\alpha_i^{(ijl)},\alpha_i^{(ikl)})}{\sin\alpha_i^{(ijk)}\sin\alpha_i^{(ijl)}\sin\alpha_i^{(ikl)}},
\end{equation}
respectively. Separately, these each have the same link to the Jacobi elliptic functions as the functions of a spherical triangle. They imply that
\begin{equation}
\sin\phi_{il}=k_1k_2\sin\theta_{jk}.
\end{equation}
Introducing uniformising variables $a_{jk}$, $j,k=1,2,3,4$ with $k>j$, associated with the six $\theta_{jk}$, such that
\begin{equation}
a_{jk}=\int_{0}^{\theta_{jk}}\frac{\mathrm{d}t}{\sqrt{(1-k_1^2\sin^2t)(1-k_1^2k_2^2\sin^2t)}},
\end{equation}
with $\theta_{jk}=\mathrm{am}(a_{jk})$, the various hyperspherical trigonometric functions can be identified with Pallewek's generalised Jacobi elliptic functions via the identifications
\begin{equation}
s(a_{jk})\equiv\sin\theta_{jk}\iff k_1s(a_{jk})=\sin\alpha_i^{(ijk)}\iff k_1k_2s(a_{jk})=\sin\phi_{il},
\end{equation}
in which, $k_1$ and $k_1k_2$ are the two moduli of the functions, and are as given earlier. From these and the identities listed previously it follows that
\begin{equation}
c(a_{jk})=\cos\theta_{jk},\; d_1(a_{jk})=\cos\alpha_i^{(ijk)},\; d_2(a_{jk})=\cos\phi_{il}.
\end{equation}
Note that under these identifications the identities (\ref{81}) for the generalised Jacobi functions are still obeyed.

The modulus $k_1$ governs the relations between the spherical trigonometry of each of the faces, and the elliptic functions. 
The second modulus, $k_1k_2$, acts as an overall modulus for the spherical tetrahedron. Note its dependence on the first 
modulus, $k_1$. These moduli remain constant, so that if one of the vertices of the spherical tetrahedron were moved, the 
others must be adjusted to compensate. These movements result in a change to the faces of the tetrahedron, but not to $k_1$, 
the facial modulus. 

\iffalse 
\noindent {\bf Remark:}
By taking various ratios of the addition formulae \eqref{85}-\eqref{88} for the generalised Jacobi elliptic 
functions we obtain some more manageable addition formulae. For example, 
\begin{equation} \label{eq:sdcd} 
\frac{s(u+v)}{c(u+v)}=\frac{s(u)d_2(u)c(v)d_1(v)+s(v)d_2(v)c(u)d_1(u)}{c(u)d_2(u)c(v)d_2(v)-(k_1k_2)'^2s(u)d_1(u)s(v)d_1(v)} 
\end{equation}  
reduces to 
\begin{equation}\label{eq:dsdc}
\frac{d_2(u)}{s(u)}\frac{c(v)}{s(v)}= 
\end{equation}
\fi 

\section{Examples}

As an example of where these hyperspherical formulae and their connection with the generalised Jacobi elliptic functions may be used,  we consider a multidimensional generalisation of the Euler top in relation to Nambu mechanics, \cite{Nambu}. We provide a link between this example and the double elliptic model.
\subsection{Multidimensional Euler Top}

The Euler top describes a free top, moving in the absence of any external torque. In their paper on Nambu Quantum mechanics, \cite{Minic}, Minic and Tze reformulated the mechanics of this top in terms of a Nambu system of order three, which they then solved in terms of the Jacobi elliptic functions. We consider a higher dimensional analogue of their example, formulating the top's mechanics in terms of a Nambu system of order four, and then solving this system in terms of the Generalised Jacobi elliptic functions.

The Euler top describing the motion in terms of the angular momentum vector 
$\boldsymbol{M}=(M_1(t),M_2(t),M_3(t)$, admits two conserved quantities: the square of the angular momentum, and the 
total energy respectively 
\begin{equation}
H_1=\frac{1}{2}\sum_{i=1}^{3}M_i^2, \quad {\rm and}\quad H_2=\frac{1}{2}\sum_{i=1}^{3}\frac{1}{I_i}M_i^2,
\end{equation} 
where $I_i$ denote the principal moments of inertia. 
The authors of \cite{Minic} take these both to be Hamiltonians for the the system, and consider, following \cite{takhtajan}, the action 
\begin{equation}\label{eq:action} 
S=\int M_1 \mathrm{d}M_2\wedge \mathrm{d}M_3 -H_1\mathrm{d}H_2\wedge \mathrm{d}t.
\end{equation}
whose critical trajectories under $\delta S=0$ yields 
the equations of motion 
\begin{equation}\label{eq:Nambu-eqs} 
\frac{dM_i}{dt}=\{M_i,H_1,H_2\},\,i=1,2,3,
\end{equation}
follow. Here $\{H_1,H_2,M_i\}$ is the Nambu 3-bracket, a ternary bracket, introduced first in \cite{Nambu}, satisfying the following axioms for Nambu brackets as given by Takhtajan, \cite{takhtajan}:
\begin{itemize}
\item
Skew-symmetry
\begin{equation}
\{A_1,A_2,A_3\}=(-1)^{\epsilon(p)}\{A_{p(1)},A_{p(2)},A_{p(3)}\},
\end{equation}
where $p(i)$ is the permutation of indices and $\epsilon(p)$ is the parity of the permutation.
\item
Derivation in each entry (the Leibniz rule)
\begin{equation}
\{A_1A_2,A_3,A_4\}=A_1\{A_2,A_3,A_4\}+\{A_1,A_3,A_4\}A_2.
\end{equation}
\item
Fundamental Identity (analogue of Jacobi identity)
\begin{equation}
\begin{split}
\left\{\left\{A_1,A_2,A_3\right\},A_4,A_5\right\}&+\left\{A_3,\left\{A_1,A_2,A_4\right\},A_5\right\}+\left\{A_3,A_4,\left\{A_1,A_2,A_5\right\}\right\}\\
&=\left\{A_1,A_2,\left\{A_3,A_4,A_5\right\}\right\}.
\end{split}
\end{equation}
\end{itemize}
Note the similarity between the Fundamental Identity and the vectorial identity (\ref{9}).
Taking $I_3>I_2>I_1$ without loss of generality, the canonical 
Nambu bracket $\{M_1,M_2,M_3\}=1$ in \eqref{eq:Nambu-eqs} 
yields the equations of motions as
\bse 
\begin{equation}
\frac{dM_1}{dt}=\left(\frac{1}{I_3}-\frac{1}{I_2}\right)M_2M_3,
\end{equation}
\begin{equation}
\frac{dM_2}{dt}=\left(\frac{1}{I_1}-\frac{1}{I_3}\right)M_1M_3,
\end{equation}
\begin{equation}
\frac{dM_3}{dt}=\left(\frac{1}{I_2}-\frac{1}{I_1}\right)M_1M_2,
\end{equation}\ese 
which have solutions of the form
\bse \begin{equation}
M_1(t)=A_1\mathrm{sn}(\lambda(t-t_0);k),
\end{equation}
\begin{equation}
M_2(t)=A_2\mathrm{cn}(\lambda(t-t_0);k),
\end{equation}
\begin{equation}
M_3(t)=A_3\mathrm{dn}(\lambda(t-t_0);k),
\end{equation}\ese 
with a velocity $\lambda$, and  
where the amplitudes can be computed in terms of the Hamiltonians $H_1$ and $H_2$ as   
\[ A_1^2=\frac{2I_1(H_1-I_3H_2)}{I_1-I_3}\ , \quad 
A_2^2=\frac{2I_2(H_1-I_3H_2)}{I_2-I_3}\ , \quad 
A_3^2=\frac{2I_3(H_1-I_2H_2)}{I_3-I_2}\  , \] 
where $H_1$ and $H_2$ are naturally constants of the motion. With the following relation for the modulus of the Jacobi elliptic functions: 
\begin{equation}
k^2=\frac{I_2-I_1}{I_3-I_1}\,\frac{H_1-I_3H_2}{H_1-I_2H_2}, 
\end{equation}
the amplitudes $A_i$ can be cast in the slightly more symmetric form: 
\[ A_1^2=\frac{-I_1^2I_2I_3}{(I_2-I_1)(I_3-I_1)}k^2\lambda^2\ , \quad 
A_2^2=\frac{I_1I_2^2I_3}{(I_1-I_2)(I_3-I_2)}k^2\lambda^2\ , \quad 
A_3^2=\frac{I_1I_2I_3^2}{(I_1-I_3)(I_2-I_3)}\lambda^2\ . \] 
%with $K:=2(I_1-I_3)(I_2H_1-H_2)=I_1I_2I_3c^2$ a constant of the motion. 
%, and where the modulus $k$ of the Jacobi elliptic functions is given by
Using Irwin's correspondence, \cite{irwin}, these solutions may be reparameterised in terms of spherical trigonometry. Note that in this case $\theta_{ij}=\alpha_k$ for all $i,j,k=1,2,3$.

A generalisation of this method can be made in higher dimensions, and in the case of a four dimensional space, we formulate the mechanics of the top in terms of a Nambu mechanical system of order four, which we then solve using Pawellek's generalised Jacobi elliptic functions discussed in the previous section. 
Consider a 4-dimensional extension of the action 
\eqref{eq:action} given by 
\begin{equation}
S=\int M_1\mathrm{d}M_2\wedge \mathrm{d}M_3\wedge \mathrm{d}M_4-H_1\mathrm{d}H_2\wedge \mathrm{d}H_3\wedge \mathrm{d}t,
\end{equation}
which describes a four-dimensional spinning top, with three 
Hamiltonians $H_1,$ $H_2$ and $H_3$, given by
\begin{equation}
H_1=\frac{1}{2}\sum_{i=1}^{4}M_i^2,
\quad 
H_2=\frac{1}{2}\sum_{i=1}^{4}\alpha_iM_i^2,
\quad 
H_3=\frac{1}{2}\sum_{i=1}^{4}\beta_iM_i^2,
\end{equation}
with $\alpha_i$ and $\beta_i$ constants. The top's equations of motion are given by the Nambu-Poisson brackets
\begin{equation}
\frac{dM_i}{dt}=\{M_i,H_1,H_2,H_3\},\quad i=1,\cdots,4\ , 
\end{equation}
where for functions $F,G,H,I$ of the variables $M_i$ we have 
\begin{equation}
\{F,G,H,I\}=\epsilon^{ijkl}\partial_{M_i}F\partial_{M_j}G\partial_{M_k}H\partial_{M_l}I.
\end{equation}
This gives the intertwined differential system of four variables
\bse\begin{equation}
\dot{M}_1=\left(\alpha_2\beta_3-\alpha_3\beta_2+\alpha_3\beta_4-\alpha_4\beta_3+\alpha_4\beta_2-\alpha_2\beta_4\right)M_2M_3M_4,
\end{equation}
\begin{equation}
\dot{M}_2=\left(\alpha_1\beta_4-\alpha_4\beta_1+\alpha_3\beta_1-\alpha_1\beta_3+\alpha_4\beta_3-\alpha_3\beta_4\right)M_1M_3M_4,
\end{equation}
\begin{equation}
\dot{M}_3=\left(\alpha_1\beta_2-\alpha_2\beta_1+\alpha_2\beta_4-\alpha_4\beta_2+\alpha_4\beta_1-\alpha_1\beta_4\right)M_1M_2M_4,
\end{equation}
\begin{equation}
\dot{M}_4=\left(\alpha_1\beta_3-\alpha_3\beta_1+\alpha_2\beta_1-\alpha_1\beta_2+\alpha_3\beta_2-\alpha_2\beta_3\right)M_1M_2M_3.
\end{equation}\ese 
These equations can be integrated in closed form in terms of the generalised Jacobi elliptic functions $s,$ $c,$ $d_1$ and $d_2,$ giving
\begin{equation}
\begin{array}{c}
M_1(t)=A_1s(\lambda(t-t_0);k_1,k_2),\\
M_2(t)=A_2c(\lambda(t-t_0);k_1,k_2),\\
M_3(t)=A_3d_1(\lambda(t-t_0);k_1,k_2),\\
M_4(t)=A_4d_2(\lambda(t-t_0);k_1,k_2),\\
\end{array}
\end{equation}
with $A_i,$ $\lambda$ and $t_0$ all constants. From \eqref{eq:Pawders} they satisfy 
\begin{eqnarray*}
\frac{A_1\lambda}{A_2A_3A_4}&=&\alpha_2\beta_3-\alpha_3\beta_2+\alpha_3\beta_4-\alpha_4\beta_3+\alpha_4\beta_2-\alpha_2\beta_4,\\
-\frac{A_2\lambda}{A_1A_3A_4}&=&\alpha_1\beta_4-\alpha_4\beta_1+\alpha_3\beta_1-\alpha_1\beta_3+\alpha_4\beta_3-\alpha_3\beta_4,\\
-\frac{k_1^2A_3\lambda}{A_1A_2A_4}&=&\alpha_1\beta_2-\alpha_2\beta_1+\alpha_2\beta_4-\alpha_4\beta_2+\alpha_4\beta_1-\alpha_1\beta_4,\\
-\frac{k_2^2A_4\lambda}{A_1A_2A_3}&=&\alpha_1\beta_3-\alpha_3\beta_1+\alpha_2\beta_1-\alpha_1\beta_2+\alpha_3\beta_2-\alpha_2\beta_3,\\
\end{eqnarray*}
and where the moduli $k_1,k_2$ can be expressed in terms of the 
conserved quantities $H_1,H_2,H_3$. 
Consequently, the motion of the top may be entirely parameterised in terms of the generalised Jacobi elliptic functions, and therefore, through the link with hyperspherical trigonometry, the angles of hyperspherical tetrahedra. Note again that this identification ensures that $\theta_{ij}=\alpha_k=\phi_{kl}$, for all $i,j,k,l=1,2,3,4$.

\subsection{Double Elliptic Systems (DELL)}
The so-called DELL, or double-elliptic, model is a conjectured 
generalisation of the Calogero-Moser and Ruijsenaars-Schneider models 
(integrable many-body systems), which is elliptic in both the momentum 
and position variables. So far, only the 2-particle model (reducing to 
one degree of freedom) has been explicitly 
constructed, \cite{braden,braden2}. A possible Hamiltonian for the three-particle model was later suggested in terms of Riemann theta 
functions, \cite{Aminov,MirMor}, supported by numerical evidence in 
\cite{MirMor}, although it remained to be proven. 

The Dell Hamiltonian for this 2-particle model,
\begin{equation}\label{eq:DELLHam} 
H(P,Q)=\alpha(Q)\mathrm{cn}\left(P\sqrt{k'^2+k^2\alpha^2(Q)};\frac{k\alpha(Q)}{\sqrt{k'^2+k^2\alpha^2(Q)}}\right),
\end{equation}
with
\begin{equation}
\alpha^2(Q)=\alpha^2_{rat}(Q)=1-\frac{2g^2}{Q^2},
\end{equation}
can easily be parameterised in terms of Pawellek's generalised Jacobi elliptic functions using the identity
\begin{equation}
\mathrm{cn}(k'_2u,\kappa)=\frac{c(u;k_1,k_2)}{d_2(u;k_1,k_2)},\quad \kappa^2=\frac{k_1^2-k_2^2}{{k'}_2^2}\  , 
\end{equation}
given by Pawellek. From this, it follows that the Hamiltonian takes the much neater form
\begin{equation}
H(P,Q)=\alpha(Q)\frac{c(P;k_1,k_2)}{d_2(P;k_1,k_2)},
\end{equation}
with
\begin{equation}
%\begin{array}{c}
k_1=k, \quad 
k_2=k\sqrt{(1-\alpha^2(Q))}.
%\end{array}
\end{equation}
The Hamilton equations of motion read
\bse\label{eq:Hameqs}
\begin{eqnarray} 
\dot{P} &=& -\frac{\partial H}{\partial Q}= -\alpha'_{ell}(Q;\widetilde{k}) \frac{c(P;k_1,k_2)}{d_2(P;k_1,k_2)}\ , \\ 
\dot{Q} &=& \frac{\partial H}{\partial P}= -{k'}_2^2\alpha'_{ell}(Q;\widetilde{k}) 
\frac{s(P;k_1,k_2)d_1(P;k_1,k_2)}{d_2^2(P;k_1,k_2)}\ , 
\end{eqnarray} 
\ese 
where 
\[ \alpha'_{ell}(Q;\widetilde{k})=\frac{d \alpha_{ell}(Q;\widetilde{k})}{dQ}=4g^2 {\rm sn}(Q;\widetilde{k})\, 
{\rm cn}(Q;\widetilde{k})\,{\rm dn}(Q;\widetilde{k})\ . 
\] 
Therefore, from this it makes sense to take a closer look at the model in terms of the generalised Jacobi elliptic functions. 

The following system of four quadrics in $\mathbb{C}^6$ provides a phase space for this two-body double elliptic system:
\begin{equation}
\begin{array}{lc}
Q_1: & x_1^2-x_2^2=1,\\
Q_2: & x_1^3-x_3^2=k^2,\\
Q_3: & -g^2x_1^2+x_4^2+x_5^2=1,\\
Q_4: & -g^2x_1^2+x_4^2+\widetilde{k}^{-2}x_6^2=\widetilde{k}^{-2},
\end{array}
\end{equation}
where $g$ is a coupling constant, and $x_i$, $i=1,2,3,4$ are affine coordinates. 
The first pair of equations provides the embedding of an elliptic curve, while the second pair 
provides a second elliptic curve which is locally fibred over the first. 
Note that when $g=0$, the system simply becomes two copies of elliptic curves embedded in $\mathbb{C}^3\times\mathbb{C}^3$. The Poisson brackets are given by 
\begin{equation}\label{eq:PBs} 
\{ x_i,x_j\} =\varepsilon_{ijk_1\cdots k_4}\frac{\partial Q_1}{\partial x_{k_1}}\frac{\partial Q_2}{\partial x_{k_2}}
\frac{\partial Q_3}{\partial x_{k_3}}\frac{\partial Q_4}{\partial x_{k_4}}\ , 
\end{equation}
with the polynomials $Q_i$ yielding the Casimirs of the algebra. The relevant Poisson brackets for this system of quadraics are then\cite{braden2}
\begin{equation}
\begin{aligned}[c]
\dot{x}_1&=\{x_1,x_5\}=x_2x_3x_4x_6,\\
\dot{x}_3&=\{x_3,x_5\}=x_1x_2x_4x_6,\\
\dot{x}_5&=\{x_5,x_5\}=0,
\end{aligned}
\qquad
\begin{aligned}[c]
\dot{x}_2&=\{x_2,x_5\}=x_1x_3x_4x_6,\\
\dot{x}_4&=\{x_4,x_5\}=x_1x_2x_3x_6,\\
\dot{x}_6&=\{x_6,x_5\}=0.
\end{aligned}
\end{equation}
These equations of motion can solved in terms of Jacobi elliptic functions, where the $x_4,x_5,x_6$ depend on 
complicated expressions for the moduli (as in \eqref{eq:DELLHam}), cf. \cite{braden2}, (we omit the formulae), 
but they are more readily solved in terms of the generalised Jacobi elliptic functions, namely by the following 
expressions: 
\begin{equation}
\begin{aligned}[c]
x_1&=A_1 s(K(t-t_0);k_1,k_2),\\
x_3&=A_3d_1(K(t-t_0);k_1,k_2),\\
x_5&=E,\,\mathrm{the\,\, energy},
\end{aligned}
\qquad
\begin{aligned}[c]
x_2&=A_2 c(K(t-t_0);k_1,k_2),\\
x_4&=A_4 d_2(K(t-t_0);k_1,k_2),\\
x_6&=K,\,\mathrm{constant},
\end{aligned}
\end{equation}
with
\begin{equation}
\begin{aligned}[c]
A_1^2&=\frac{k_1k_2}{\sqrt{-g^2}},\\
A_3^2&=\frac{-k_2}{k_1\sqrt{-g^2}},
\end{aligned}
\qquad
\begin{aligned}[c]
A_2^2&=\frac{-k_1k_2}{\sqrt{-g^2}},\\
A_4^2&=\frac{k_1\sqrt{-g^2}}{k_2},
\end{aligned}
\end{equation}
respectively, identifying $k_1=1/\widetilde{k}$ and $k_2=\widetilde{k}\sqrt{-2g^2}$. 
Note that $K$ is related to the energy, $E$, through the final two quadrics
\begin{equation}
\tilde{k}^2(1-E^2)=1-K^2.
\end{equation}
\iffalse 
Substituting these solutions into the quadrics and comparing the relations gives solutions
\begin{equation}
\begin{aligned}[c]
x_1&=s(\gamma t),\\
x_3&=\imath kd_1(\gamma t),\\
x_5&=E,
\end{aligned}
\qquad
\begin{aligned}[c]
x_2&=\imath c(\gamma t),\\
x_4&=\frac{d_2(\gamma t)}{k},\\
x_6&=\gamma.
\end{aligned}
\end{equation}
\fi 
This connection through the quadrics suggests that the generalised Jacobi functions form the natural parametrisation for the 
2-particle DELL model.  
Note that these equations of motion also coincide with those for the four-dimensional Euler top example of the 
previous subsection. Hence, these two systems must be equivalent up to scaling.
 
\section{Conclusion}

We have investigated a novel connection between the formulae governing 4-dimensional hyperspherical geometry and elliptic functions. This connection is via the generalised Jacobi elliptic functions, defined through Abelian integrals associated with a doubly elliptic cover of an elliptic curve. These generalised Jacobi elliptic functions have two distinct moduli, $k_1$ and $k_2$, and are expressible in terms of the usual Jacobi elliptic functions, . As such these functions may be of interest for the study of certain elliptic integrable models where solutions in terms of elliptic functions with different moduli appear, for example the Q4 lattice equation of the ABS classification, and the DELL model. The link between addition formulae and hyperspherical geometry suggests that discrete integrable models, such as discretisation of the DELL model, may be derived exploiting this connection. This is a link that has been exploited in the spherical case by Petrera and Suris, \cite{Suris}, in producing an integrable map in the sense of multidimensional consistency, based upon the cosine and polar cosine rules. They have shown this map to be related to the Hirota-Kimura discretisation of the Euler top. Whilst this touches slightly on our work, we are more concerned with the higher dimensional hyperspherical trigonometry, and its links to elliptic functions. The Dell model, which arises in our case, although equivalent to a higher dimensional Euler top, seems to be different from Petrera and Suris' model.

We have also generalised the nested structure of hyperspherical trigonometry here to any dimension using higher-dimensional vector products, although the inter-relations between the formulae become increasingly complicated.

In recent years there has been substantial progress in obtaining 
formulas for hyperspherical tetrahedra in 4D in terms of arc lengths 
and dihedral angles, cf. \cite{Murakami,Murakami2,Murakami3}. 
These generalise results for symmetric tetrahedra obtained earlier,  
\cite{Derevnin,Derevnin2}, refining the formula for the general case 
found in \cite{ChoKim}. The expressions involve Sch\"afli type 
symbols, and can be expressed in terms of dilogarithms (Lobachevsky functions in the hyperbolic case).  The occurrence of such functions 
in Lagrangians of certain discrete integrable systems, cf. e.g. \cite{Lobb}, 
seems to point to suggest a connection between the two types 
objects from a geometrical perspective. 

Finally, the occurrence of the Pawellek functions, which involve 
elliptic functions of two distinct moduli, is reminiscent of the 
\textit{bi-elliptic addition formulae}, \cite{AHN}, that occur naturally in the 
seed and soliton solutions of the Q4 lattice equation, 
discovered by Adler, \cite{Adler}.

\subsection*{Acknowledgements}

The first author was supported by the UK Engineering and Physical Sciences Research Council (ESPRC) when this work was conducted.
Part of the work was done while the second author was supported by a Royal Society/Leverholme Trust Senior Research Fellowship (2011-12). 
%He is currently supported by the EPSRC grant EP/W007290/1. 

\appendix

\section{Four and Five Parts Formulae}
\def\theequation{A\arabic{equation}}
\setcounter{equation}{0}

We write a general position vector on the surface of a sphere in terms of its basis vectors, and derive the four and five parts formulae.

Given 2 vectors in three dimensions, we can express any other vector on the sphere in terms of these by
\begin{equation}
\mathbf{n}_i=A\mathbf{n}_j+B\mathbf{n}_k+C\mathbf{n}_j\times\mathbf{n}_k,
\end{equation}
with
\begin{equation}
\left(\begin{array}{c}
A\\
B\\
\end{array}\right)
=\left(\begin{array}{cc}
1 & \cos\theta_{jk}\\
\cos\theta_{jk} & 1\\
\end{array}\right)^{-1}
\left(\begin{array}{c}
\cos\theta_{ij}\\
\cos\theta_{ik}\\
\end{array}\right)
=\frac{1}{\sin^2\theta_{jk}}\left(\begin{array}{c}
\sin\theta_{ik}\sin\theta_{jk}\cos\alpha_k\\
\sin\theta_{ij}\sin\theta_{jk}\cos\alpha_j\\
\end{array}\right),
\end{equation}
and
\begin{equation}
C=\frac{\mathbf{n}_i\cdot(\mathbf{n}_j\times\mathbf{n}_k)}{\sin^2\theta_{jk}}=\frac{\sin\alpha_i\sin\theta_{ik}\sin\theta_{ij}}{\sin^2\theta_{jk}},
\end{equation}
giving
\begin{equation}
\mathbf{n}_i=\frac{\sin\theta_{ik}\cos\alpha_k}{\sin\theta_{jk}}\mathbf{n}_j+\frac{\sin\theta_{ij}\cos\alpha_j}{\sin\theta_{jk}}\mathbf{n}_k+
\frac{\sin\alpha_i\sin\theta_{ik}\sin\theta_{ij}}{\sin^2\theta_{jk}}\mathbf{n}_j\times\mathbf{n}_k.
\end{equation}
Using the spherical sine rule this reduces to
\begin{equation}
\mathbf{n}_i=\frac{\sin\alpha_{j}\cos\alpha_k}{\sin\alpha_{i}}\mathbf{n}_j+\frac{\sin\alpha_{k}\cos\alpha_j}{\sin\alpha_{i}}\mathbf{n}_k+
\frac{\sin\alpha_i\sin\sin\alpha_k}{\sin\alpha_i}\mathbf{n}_j\times\mathbf{n}_k.
\end{equation}
Using the Gram-Schmidt Process $\mathbf{n}_i$ may be written in terms of an orthonormal basis,
\begin{equation}
\mathbf{n}_i=A'\mathbf{n}_{j}+B'\mathbf{n}_{k}'+C'\mathbf{n}_{j}\times\mathbf{n}_{k}',
\end{equation}
where $\mathbf{n}_j$ is as is before, and $\mathbf{n}_k'$ is given by
\begin{equation}
\mathbf{n}_k'=\frac{\mathbf{n}_k-(\mathbf{n}_j\cdot\mathbf{n}_k)\mathbf{n}_j}{|\mathbf{n}_k-(\mathbf{n}_j\cdot\mathbf{n}_k)\mathbf{n}_j|}
=\frac{\mathbf{n}_k-(\mathbf{n}_j\cdot\mathbf{n}_k)\mathbf{n}_j}{\sin\theta_{jk}}.
\end{equation}
This implies
\begin{equation}
C'=(\sin\theta_{jk})C,\hspace{15mm}B'=(\sin\theta_{jk})B,\hspace{15mm}A'=A+\cot\theta_{jk}B',
\end{equation}
from which it follows that
\begin{equation}
\mathbf{n}_i=\frac{\sin\theta_{ik}\cos\alpha_k+\cos\theta_{jk}\sin\theta_{ij}\cos\alpha_j}{\sin\theta_{jk}}\mathbf{n}_{j}+\sin\theta_{ij}
\cos\alpha_j\mathbf{n}_{k}'+\sin\theta_{ij}\sin\alpha_j\mathbf{n}_{j}\times\mathbf{n}_{k}'.
\end{equation}
Equating the two expression we have for $\mathbf{n_i}$ gives the four-parts formula
\begin{equation}
\cot\theta_{ij}\sin\theta_{jk}=\cot\alpha_k\sin\alpha_j+\cos\theta_{jk}\cos\alpha_j.
\end{equation}
Similarly, for 
\begin{equation}
\mathbf{u}_{ij}=P\mathbf{u}_{ik}+Q\mathbf{u}_{kj}+R\mathbf{u}_{ik}\times\mathbf{u}_{kj}
\end{equation}
it follows that
\begin{equation}
\mathbf{u}_{ij}=\frac{\sin\alpha_j}{\sin\alpha_k}\cos\theta_{jk}\mathbf{u}_{ik}
+\frac{\sin\alpha_i}{\sin\alpha_k}\cos\theta_{ik}\mathbf{u}_{kj}-
\frac{\sin\theta_{ik}\sin\theta_{kj}}{\sin\theta_{ij}}\mathbf{u}_{ik}\times\mathbf{u}_{kj},
\end{equation}
or alternatively, using the spherical sine rule,
\begin{equation}
\mathbf{u}_{ij}=\frac{\sin\theta_{ik}\cos\theta_{jk}}{\sin\theta_{ij}}\mathbf{u}_{ik}
+\frac{\sin\theta_{jk}\cos\theta_{ik}}{\sin\theta_{ij}}\mathbf{u}_{kj}-
\frac{\sin\theta_{ik}\sin\theta_{kj}}{\sin\theta_{ij}}\mathbf{u}_{ik}\times\mathbf{u}_{kj}.
\end{equation}
Similarly, using the Gram-Schmidt Process,
\begin{equation}
\mathbf{u}_{ij}=\frac{\sin\alpha_{j}\cos\theta_{jk}-\cos\alpha_k\sin\alpha_i\cos\theta_{ik}}{\sin\alpha_{k}}\mathbf{u}_{ik}+\sin\alpha_i\cos\theta_{ik}\mathbf{u}_{kj}'+\sin\theta_{kj}\sin\alpha_j\mathbf{u}_{ij}\times\mathbf{u}_{kj}'.
\end{equation}
Equating the two expression we have for $\mathbf{u}_{ij}$ gives the five-parts formula
\begin{equation}
\cot\theta_{jk}\sin\alpha_{j}=\cos\alpha_i\sin\alpha_k+\cos\theta_{ik}\sin\alpha_i\cos\alpha_k.
\end{equation}

In 4-dimensional Euclidean space, for the vectors  $\mathbf{n}_i,\mathbf{n}_j,\mathbf{n}_k,\mathbf{n}_l\in E_4$ of a hyperspherical tetrahedron we have the 
expansions: 
\begin{eqnarray} 
\mathbf{n}_i&=&
\frac{\sin(\theta_{ik},\theta_{il},\theta_{kl})\,\cos\phi_{kl}}
{\sin(\theta_{jk},\theta_{jl},\theta_{kl})}\mathbf{n}_j + 
\frac{\sin(\theta_{ij},\theta_{il},\theta_{jl})\,\cos\phi_{jl}}
{\sin(\theta_{jk},\theta_{jl},\theta_{kl})}\mathbf{n}_k + 
\frac{\sin(\theta_{ij},\theta_{ik},\theta_{jk})\,\cos\phi_{jk}}
{\sin(\theta_{jk},\theta_{jl},\theta_{kl})}\mathbf{n}_l +   \\ 
&& + \frac{\sqrt{
\sin(\theta_{ij},\theta_{ik},\theta_{jk})\,\sin(\theta_{ij},\theta_{il},\theta_{jl})\, 
\sin(\theta_{ik},\theta_{il},\theta_{kl})\,\sin(\phi_{ij},\phi_{ik},\phi_{il})}}
{\left(\sin(\theta_{jk},\theta_{jl},\theta_{kl})\right)^2}
\,\mathbf{n}_j\times\mathbf{n}_k\times\mathbf{n}_l\   
\end{eqnarray} 
as well as 
\begin{eqnarray}
\mathbf{n}_i&=&\cos\theta_{ij}\,\mathbf{n}_j + \sin\theta_{ij}\cos\alpha^{(ijk)}_j\, 
\mathbf{n}^\prime_k + \sin\theta_{ij}\sin\alpha^{(ijk)}_j\cos\phi_{jk}\,
\mathbf{n}_l^\prime   \\ 
&& ~~~~~ + \sin\theta_{ij}\sin\alpha^{(ijk)}_j\sin\phi_{jk}\,\mathbf{n}_j
\times\mathbf{n}^\prime_k\times\mathbf{n}^\prime_l\   , 
\end{eqnarray}  
in a (Gram-Schmidt) orthogonal basis $\mathbf{n}_j,\mathbf{n}'_k,\mathbf{n}'_l$. Comparing the results provides the following set of equations: 
\begin{equation} 
\begin{split}
\frac{\sin(\theta_{jk},\theta_{jl},\theta_{kl})}
{\sin(\theta_{ij},\theta_{ik},\theta_{jk})} \cos\phi_{jk} = 
&\cos\theta_{il} \Bigg( 1 - \tan\theta_{il}\cot\theta_{ik}\cos\phi_{ij} 
\frac{\sin\alpha_i^{(ijl)}}{\sin\alpha_i^{(ijk)}}   \\ 
&\ - \tan\theta_{il}\cot\theta_{ij}\cos\alpha_i^{(ijl)} 
\left( 1 - \tan\alpha_i^{(ijl)}\cot\alpha_i^{(ijk)}\cos\phi_{ij}\right) 
\Bigg)\   ,    \\  
\frac{\sin(\theta_{ik},\theta_{il},\theta_{kl})}
{\sin(\theta_{ij},\theta_{ik},\theta_{jk})} \cos\phi_{ik} =& 
\frac{\sin\theta_{il}\cos\alpha_i^{(ijl)}}{\sin\theta_{ij}} 
\left( 1 - \tan\alpha_i^{(ijl)}\cot\alpha_i^{(ijk)}\cos\phi_{ij}\right)\   ,  
\end{split}
\end{equation} 
which form the hyperspherical analogue of the four-parts formula. 

Similarly, for the expansions of the polar vectors $\mathbf{u}_{ijk},\mathbf{u}_{ijl},\mathbf{u}_{ikl},\mathbf{u}_{jkl}$ we have
\begin{equation}
\begin{split}
\mathbf{u}_{ijk}=&\frac{\sin(\theta_{ij},\theta_{il},\theta_{jl})}{\sin(\theta_{ij},\theta_{ik},\theta_{jk})}\cos\theta_{kl}\mathbf{u}_{ijl}
-\frac{\sin(\theta_{jk},\theta_{jl},\theta_{kl})}{\sin(\theta_{ij},\theta_{ik},\theta_{jk})}\cos\theta_{il}\mathbf{u}_{jlk}
+\frac{\sin(\theta_{ik},\theta_{il},\theta_{kl})}{\sin(\theta_{ij},\theta_{ik},\theta_{jk})}\cos\theta_{jl}\mathbf{u}_{lki}\\
&+\frac{\sin(\theta_{ij},\theta_{il},\theta_{jl})\sin(\theta_{jk},\theta_{jl},\theta_{kl})\sin(\theta_{ik},\theta_{il},\theta_{kl})}{\sin(\theta_{ij},\theta_{ik},\theta_{il},\theta_{jk},\theta_{jl},\theta_{kl})}\mathbf{u}_{ijl}\times\mathbf{u}_{jlk}\times\mathbf{u}_{kli}
\end{split} 
\end{equation}
and in terms of an orthonormal basis,
\begin{equation}
\begin{split}
\mathbf{u}_{ijk}=&\Bigg(\frac{\sin(\theta_{ij},\theta_{il},\theta_{jl})\cos\theta_{kl}}{\sin(\theta_{ij},\theta_{ik},\theta_{jk})}-\left(\frac{\cos\phi_{jk}-\cos\phi_{jl}\cos\phi_{ij}}{\sin\phi_{jl}}\right)\cos\phi_{jl}\\
&\hspace{1in}+\frac{\cos\theta_jl\cos\phi_{il}\sin\phi_{jl}\sin(\phi_{ij},\phi_{jk},\phi_{jl})}{\sin(\phi_{il},\phi_{jl},\phi_{kl})}\Bigg)\mathbf{u}_{ijl}\\
&+\frac{\cos\phi_{jk}-\cos\phi_{ij}\cos\phi_{jl}}{\sin\phi_{jl}}\mathbf{u}_{jlk}'+\frac{\cos\theta_{jl}\sin(\phi_{ij},\phi_{jk},\phi_{jl})}{\sin\phi_{ij}}\mathbf{u}_{lki}'\\
&+\sin(\theta_{ij},\theta_{ik},\theta_{il},\theta_{jk},\theta_{jl},\theta_{kl})\mathbf{u}_{ijl}\times\mathbf{u}_{jlk}'\times\mathbf{u}_{kli}'.
\end{split}
\end{equation}
Comparing the results provides
\begin{equation}
\begin{split}
\cos\phi_{jk}\cos\phi_{jl}\sin(\theta_{ij},\theta_{ik},\theta_{jk})=&\cos\phi_{ij}\sin\phi_{jl}\sin(\theta_{ij},\theta_{ik},\theta_{jk})\\
&+\cos\theta_{kl}\sin\phi_{jl}\sin(\theta_{ij},\theta_{il},\theta_{jl})\\
&+\cos^2\phi_{jl}\cos\phi_{ij}\sin(\theta_{ij},\theta_{ik},\theta_{jk})\\
&+\cos\theta_{jl}\cos\phi_{il}\sin^2\phi_{jl}\sin(\theta_{ik},\theta_{il},\theta_{kl}),
\end{split}
\end{equation}
the hyperspherical analogue of the five-parts formula.

\section{Angle Addition Formulas}
\def\theequation{B\arabic{equation}}
\setcounter{equation}{0}

We derive the cosine addition formulae by collapsing a spherical triangle. We then extend this to see what happens when we collapse a hyperspherical tetrahedron. Recall that the generalised sine function of three variables gives the volume of a three dimensional parallelepiped defined by vectors $\mathbf{n}_i,\mathbf{n}_j$ and $\mathbf{n}_k$ embedded in four dimensional Euclidean space,
\begin{equation}
\sin(\theta_{ij},\theta_{ik},\theta_{jk})=\left|\begin{array}{ccc}
1&
\cos\theta_{ij}&
\cos\theta_{ik}\\
\cos\theta_{ij}&
1&
\cos\theta_{jk}\\
\cos\theta_{ik}&
\cos\theta_{jk}&
1\\
\end{array}\right|^\frac{1}{2}.
\end{equation}
When the three vectors $\mathbf{n}_i,\mathbf{n}_j$ and $\mathbf{n}_k$ become coplanar, the volume of the parallelepiped collapses to zero, and hence, so does the generalised sine function. When this occurs
\begin{equation}
\sin^2(\theta_{ij},\theta_{ik},\theta_{jk})=0.
\end{equation}
By expanding this out and completing the square in terms of $\cos\theta_{ij}$, this reduces to
\begin{equation}
(\cos\theta_{ij}-\cos\theta_{ik}\cos\theta_{jk})^2=\sin^2\theta_{ik}\sin^2\theta_{jk}.
\end{equation}
Solving this gives
\begin{equation}
\cos\theta_{ij}=\cos\theta_{ik}\cos\theta_{jk}\pm\sin\theta_{ik}\sin\theta_{jk}.
\end{equation}
This occurs when $\theta_{ij}=\theta{ik}\mp\theta{jk}$, hence giving us the standard addition formula for cosine,
\begin{equation}
\cos(A\pm B)=\cos A\cos B\mp\sin A\sin B.
\end{equation}

Similarly, the generalised sine function of six variables gives the volume of a four dimensional parallelotope defined by vectors $\mathbf{n}_i,\mathbf{n}_j,\mathbf{n}_k$ and $\mathbf{n}_l$ embedded in five dimensional Euclidean space,
\begin{equation}
\sin(\theta_{ij},\theta_{ik},\theta_{il},\theta_{jk},\theta_{jl},\theta_{kl})=\left|\begin{array}{cccc}
1&
\cos\theta_{ij}&
\cos\theta_{ik}&
\cos\theta_{il}\\
\cos\theta_{ij}&
1&
\cos\theta_{jk}&
\cos\theta_{jl}\\
\cos\theta_{ik}&
\cos\theta_{jk}&
1&
\cos\theta_{kl}\\
\cos\theta_{il}&
\cos\theta_{jl}&
\cos\theta_{kl}&
1\\
\end{array}\right|^\frac{1}{2}.
\end{equation}
When the four vectors $\mathbf{n}_i,\mathbf{n}_j,\mathbf{n}_k$ and $\mathbf{n}_l$ become linearly dependent, the volume of the 4-parallelepiped collapses to zero, and hence, so does the generalised sine function. When this occurs, recalling the hyperspherical sine rule,
\begin{equation}
\sin^2(\theta_{ij},\theta_{ik},\theta_{il},\theta_{jk},\theta_{jl},\theta_{kl})=0\iff k_H=0.
\end{equation}
This gives
\begin{equation}\left.
\begin{array}{c}
\sin(\phi_{ij},\phi_{ik},\phi_{il})=0\\
\sin(\phi_{ij},\phi_{jk},\phi_{jl})=0\\
\sin(\phi_{ik},\phi_{jk},\phi_{kl})=0\\
\sin(\phi_{il},\phi_{jl},\phi_{kl})=0\\
\end{array}
\right\},
\end{equation}
which, in turn from the spherical case implies
\begin{equation}\left.
\begin{array}{c}
\phi_{ij}+\phi_{ik}+\phi_{il}=0\\
\phi_{ij}+\phi_{jk}+\phi_{jl}=0\\
\phi_{ik}+\phi_{jk}+\phi_{kl}=0\\
\phi_{il}+\phi_{jl}+\phi_{kl}=0\\
\end{array}
\right\},
\end{equation}
and hence,
\begin{equation}
\phi_{ij}+\phi_{ik}+\phi_{il}+\phi_{jk}+\phi_{jl}+\phi_{kl}=0,
\end{equation}
the sum of the dihedral angles is zero.

%\begin{thebibliography}{99}

%\bibitem{Levi1}
%Levi D and Yamilov R I, The generalized symmetry method for discrete
%equations, {\it J. Phys. A: Math. Theor.} {\bf 42}, 454012 (18pp), 2009.
%
%
%\bibitem{Levi2}
%Levi D and Yamilov R I, Generalized Lie symmetries for Difference
%Equations. In: {\it Symmetries and Integrability of Difference Equations}. Ed.
%by D. Levi, P. Olver, Z. Thomova, and P. Winternitz. London Mathematical
%Society Lecture Notes series. Cambridge: Cambridge University Press, 160–190, 2011.

\label{lastpage}
\end{document}